\documentclass[11pt, reqno,oneside]{amsart}

\usepackage{amssymb}
\usepackage{amsmath}
\usepackage{graphicx}
\usepackage{enumerate}
\usepackage{color}
\usepackage[a4paper, hdivide={2.5cm,,3cm},vdivide={3.9cm,,3.9cm}]{geometry}

\makeatletter\@addtoreset{equation}{section}\makeatother

\allowdisplaybreaks

\begin{document}

\theoremstyle{plain}
\newtheorem{theorem}{Theorem}[section]
\newtheorem{proposition}[theorem]{Proposition}
\newtheorem{corollary}[theorem]{Corollary}
\newtheorem{condition}[theorem]{Condition}
\newtheorem{example}[theorem]{Example}
\newtheorem{lemma}[theorem]{Lemma}
\theoremstyle{definition}
\newtheorem{definition}[theorem]{Definition}
\theoremstyle{remark}
\newtheorem{remark}[theorem]{Remark}
\newtheorem{notation}[theorem]{Notation}

\newcommand{\grad}{\nabla}
\newcommand{\D}{\partial}
\newcommand{\E}{\mathcal{E}}
\newcommand{\N}{\mathbb{N}}
\newcommand{\dom}{\mathcal{D}}
\newcommand{\ess}{\operatorname{ess~inf}}

\title[An Invariance Principle]{An Invariance Principle for the Tagged Particle Process in Continuum with Singular Interaction Potential}

\author{Florian Conrad, Torben Fattler and Martin Grothaus}

\address{
Florian Conrad, Department of Mathematics, University of Kaiserslautern, P.O.Box 3049, 67653
Kaiserslautern, Germany and Mathematics Department, Bielefeld University, P.O.Box 100131, 33501 Bielefeld, Germany.
{\rm \texttt{Email: fconrad@math.uni-bielefeld.de}}\newline
\indent Torben Fattler, Department of Mathematics, University of Kaiserslautern, P.O.Box 3049, 67653
Kaiserslautern, Germany.
{\rm \texttt{Email: fattler@mathematik.uni-kl.de}}\newline
\indent Martin Grothaus, Department of Mathematics, University of Kaiserslautern, P.O.Box 3049, 67653
Kaiserslautern, Germany.
{\rm \texttt{Email: grothaus@mathematik.uni-kl.de}}
}

\begin{abstract}
We consider the dynamics of a tagged particle in an infinite particle environment moving according to a stochastic gradient dynamics. For singular interaction potentials this tagged particle dynamics was constructed first in \cite{FaGr08}, using closures of pre-Dirichlet forms which were already proposed in \cite{GP85} and \cite{Os98}. The environment dynamics and the coupled dynamics of the tagged particle and the environment were constructed separately. Here we continue the analysis of these processes: Proving an essential m-dissipativity result for the generator of the coupled dynamics from \cite{FaGr08}, we show that this dynamics does not only contain the environment dynamics (as one component), but is, given the latter, the only possible choice for being the coupled process. Moreover, we identify the uniform motion of the environment as the reversed motion of the tagged particle. (Since the dynamics are constructed as martingale solutions on configuration space, this is not immediate.) Furthermore, we prove ergodicity of the environment dynamics, whenever the underlying reference measure is a pure phase of the system. Finally, we show that these considerations are sufficient to apply \cite{DeM89} for proving an invariance principle for the tagged particle process. We remark that such an invariance principle was studied before in \cite{GP85} for smooth potentials, and shown by abstract Dirichlet form methods in \cite{Os98} for singular potentials.

Our results apply for a general class of Ruelle measures corresponding to potentials possibly having infinite range, a non-integrable singularity at $0$ and a nontrivial negative part, and fulfill merely a weak differentiability condition on $\mathbb R^d\setminus\{0\}$.
\end{abstract}


\maketitle

\newcommand{\maggi}[1]{{\textcolor[rgb]{1,0,0}{*}}\marginpar{\tiny{#1}}}
\newcommand{\maggiI}[1]{\maggi{#1}}
\newcommand{\R}{{\mathbb R}}
\newcommand{\C}{{\mathbb C}}
\newcommand{\K}{{\mathbb K}}
\renewcommand{\N}{{\mathbb N}}
\newcommand{\Z}{{\mathbb Z}}
\newcommand{\Q}{{\mathbb Q}}

\section{Introduction}

In \cite{DeM89}, de Masi, Goldstein, Ferrari and Wick developed, based on ideas by Kipnis and Varadhan (see \cite{KV86}), a general method for proving an invariance principle for antisymmetric functions of reversible Markov processes and gave, among other examples, a sketch of an application of this method to the tagged particle process corresponding to a stochastic gradient dynamics, which describes a system of stochastically perturbed particles interacting via a pair potential $\phi$ (which they assumed to be $C^2$ with compact support). 

The stochastic gradient dynamics of particles at positions $x^i_t\in\R^d$, $i\in\N\cup\{0\}$, $d\in\N$, at time $t\geq 0$, from which one starts is described by the following equations:
\begin{equation}\label{eqn:sgd}
dx^i_t=-\sum_{\stackrel{\scriptstyle{j=1}}{j\neq i}}^\infty \nabla\phi(x^i_t-x^j_t)\,dt+\sqrt{2} dB_t^i,\quad i\in\N\cup\{0\},
\end{equation}
where $(B_t^i)_{t\geq 0}$, $i\in\N$, are independent $\R^d$-valued Brownian motions. It is a non-trivial question whether a solution in the sense of a sequence of random paths $(x^i_t)_{t\geq 0}$ solving the above equations can be constructed for a large class of initial configurations. For constructions of solutions in this sense see e.g.~\cite{La77a}, \cite{Fr87}. However, these constructions require either $d\leq 2$ or rather restrictive assumptions on $\phi$ (including compact support and the absence of singularities). On the other hand, the existence of Dirichlet form solutions or even martingale solutions on the space of locally finite configurations on $\R^d$ has been shown in great generality, see e.g.~\cite{Yo96}, \cite{Os96}, \cite{AKR98b}. Considering the tagged particle process means to focus on the trajectory of one particular particle and considering the rest of the system as its environment. Ignoring the question of existence of solutions of \eqref{eqn:sgd}, this can be done by denoting the position of the tagged particle at time $t$ by $\xi_t:=x_t^0$, and the environment system by $y_t^i:=x_t^i-x_t^0$, $i\in\N$. The thereby obtained joint equations of the tagged particle and its environment read as follows:
\begin{align}
&d\xi_t=\sum_{i=1}^\infty\nabla\phi(y^i_t)\,dt+\sqrt{2}\,dB^0_t,\label{eqn:tag}\\
&\left.
\begin{matrix}
dy^i_t=-\sum_{\stackrel{\scriptstyle{j=1}}{j\not=i}}^\infty\nabla\phi(y^i_t-y^j_t)-\nabla\phi(y^i_t)-\sum_{j=1}^\infty\nabla\phi(y^j_t)\,dt\\
+\sqrt{2}\,d(B^{i}_t-B^0_t),   
\end{matrix}\right\},\quad i\in\N,\label{eqn:env}
\end{align}
$t\geq 0$. Again, the construction schemes from \cite{La77a}, \cite{Fr87} may be applied to construct directly solutions to this system, as is mentioned in \cite{GP85}, \cite{DeM89}. However, of course, the same restrictions apply as in the case of the system \eqref{eqn:sgd}. For a physically realistic setting (singular potentials, e.g.~of Lennard-Jones type, arbitrary dimension, e.g.~$d=3$) the first construction of solutions in the sense of the martingale problem on the space of locally finite configurations was given in \cite{FaGr08}. Let us make the solution concept more precise. The process constructed in \cite{FaGr08} has values in $\R^d\times\Gamma$, where
$$
\Gamma:=\{\gamma\subset\R^d\,|\,\sharp(\gamma\cap \Lambda)<\infty\mbox{ for all bounded $\Lambda\subset\R^d$}\},
$$
is the space of simple locally finite configurations in $\R^d$ (for $d=1$ the state space is slightly larger). Here $\R^d$ represents the position of the tagged particle and the elements of $\Gamma$ replace sequences $(y^i)_{i\in\N}$ of positions of the particles in the environment. That the process solves \eqref{eqn:tag}, \eqref{eqn:env} in the sense of the martingale problem means that it solves the martingale problem for the generator corresponding to these equations which is (at first) defined on $C_0^\infty(\R^d)\otimes \mathcal FC_b^\infty(C_0^\infty(\R^d),\Gamma)$, i.e.~functions which are $C_0^\infty$ in the component of the tagged particle and smooth functions depending on the particles from a bounded region in space. (In fact, the martingale problem is solved for all functions in the domain of the Friedrichs extension of this generator in a suitable $L^2$-space. For more details see Section \ref{sec:existence}.) It is not clear whether these solutions lead to weak solutions in the ordinary sense (it is e.g. not clear whether one can distinguish and enumerate the particles or whether particles enter from infinity etc.). In particular, any result one may obtain for this type of solutions has to be derived without using the equations \eqref{eqn:tag}, \eqref{eqn:env}, but instead the semimartingale decompositions for functions from the domain of the generator. The construction in \cite{FaGr08} is done with the help of a quasi-regular Dirichlet form defined on $L^2(\R^d\times\Gamma;dx\otimes \mu)$, where $\mu\in\mathcal G_{ibp}^{gc}(\Phi_\phi,ze^{-\phi}dx)$, the set of grand canonical Gibbs measures for the interaction potential $\phi$ with intensity measure given by $z e^{-\phi}dx$ for some $z>0$, such that the correct integration by parts formulae are valid; roughly spoken the latter means that elements of $\mathcal G_{ibp}^{gc}(\Phi_\phi,ze^{-\phi}dx)$ are ``translation invariant except of the nonconstant intensity'' (again see the discussion in Section \ref{sec:existence}). In addition to this ``coupled process'' (consisting of the joint dynamics of the tagged particle and the environment) in \cite{FaGr08} an ``environment process'' is constructed from another Dirichlet form on $L^2(\Gamma;\mu)$, which gives a martingale solution only of \eqref{eqn:env} in the above sense. (The connection between both Dirichlet forms and processes is a priori not clear, but can be established, see below.)

\vspace{2ex}Our aim is to apply the general result on invariance principles from \cite{DeM89} to the martingale solutions from \cite{FaGr08} by adapting it to this type of solutions. A detailed understanding of the dynamics, in particular, the precise relation between the path of the tagged particle and its environment, is essential for doing so. 

We obtain the applicability of \cite{DeM89} by answering a number of natural questions about solutions constructed in \cite{FaGr08}. Let us at first explain these questions, which are interesting in themselves, and then give the connection to what we need for the invariance principle.

The first question is the following: We already mentioned above that it seems at least difficult to get from martingale solutions to weak solutions of the above equations. However, it should be possible to verify at least the validity of the $d$-dimensional equation \eqref{eqn:tag}. This consists of two tasks: First, one has to verify that this stochastic equation is indeed fulfilled with \emph{some} Brownian motion $(B_t^0)_{t\geq 0}$. Then one has to make sure that this Brownian motion ``fulfills'' \eqref{eqn:env}, which means here that it should coincide with the stochastic part of the uniform motion of the particles in the environment. As is mentioned in \cite{DeM89}, if one knows that \eqref{eqn:env} and \eqref{eqn:tag} are fulfilled, one can recover $B_t^0$ from the environment process by computing $-\lim_{N\to\infty} \frac{1}{N}\sum_{i=1}^N (B_t^i-B_t^0)$, showing in particular that the displacement of the tagged particle is a functional of the environment. In Section \ref{sec:displacement} we show that this can also be done when starting with semimartingale decompositions of functions from $C_0^\infty(\R^d)\times\mathcal FC_b^\infty(C_0^\infty(\R^d),\Gamma)$ and thus one does not need to have a solution of the above equations in a stronger sense than in the sense of the martingale problem.
Recently, we became aware of the article \cite{Os11}. Therein the question of the validity of equation (\ref{eqn:tag}) is discussed. Osada develops a new approach to solve infinite dimensional stochastic differential equations for a general class of potentials. The author focuses in the application on 2D Coulomb potentials, but interaction potentials from the Ruelle class might also be treated, see \cite[Example 21]{Os11}.

The second question is the question how the environment process and the environment part (second component) of the coupled process from \cite{FaGr08} are related. Since the construction of the environment process is done from another Dirichlet form (the connection of which with the Dirichlet form of the coupled process is not clear) at first sight the environment part of the coupled process does not need to coincide with this environment process. That the processes are indeed the same, is verified in \cite{Os98}, see the proof of Lemma 2.3 therein. Nevertheless, we give another proof of this fact by means of a stronger result: We prove an $L^q$-uniqueness property (i.e.~essential m-dissipativity result) of the generator of the coupled process on the set $C_0^\infty(\R^d)\otimes D(L_{env})_b$, where $D(L_{env})_b$ denotes the set of bounded functions from the domain of the generator of the environment process. One should mention that this domain is rather large, since it includes somehow ``full'' knowledge of the environment part of the process. Nevertheless, it is not trivial to prove it, because on a technical level (in terms of generators) one sees that the coupling between the tagged particle and the environment as well as the drift acting on the tagged particle is not a small perturbation of the stochastic motion of the environment and the tagged particle. The $L^q$-uniqueness property does not only give sufficient knowledge of the semigroup corresponding to the coupled process in order to identify the environment process as the environment component of the coupled process (and thus essentially reprocudes the mentioned result from \cite{Os98}). It also shows that for the environment process from \cite{FaGr08} there exists exactly one possible tagged particle process: it is not possible to construct another ``coupled process'' having the same environment part and solving the martingale problem for the generator corresponding to \eqref{eqn:tag}, \eqref{eqn:env}. The $L^q$-uniqueness result and the identification of the environment component of the coupled process are contained in Section \ref{sec:uniqueness}.

The third question is simply the question whether the environment process is ergodic. For the stochastic gradient dynamics an answer is given in \cite{AKR98b}: If the grand canonical reference measure $\mu$ for the Dirichlet form is an extremal canonical Gibbs measure, ergodicity is shown there in great generality. Although the approach given in that paper should also work in the present setting, it seems more effective to make use of the fact that contrary to the case of the stochastic gradient dynamics, the environment process contains also a non-degenerate uniform stochastic motion. In fact, we find that under assumptions which (except possibly in very exceptional cases and in the one-dimensional setting with a singular potential) contain the ones from \cite{AKR98b}, we obtain ergodicity from this uniform motion, if $\mu$ is extremal in $\mathcal G_{ibp}^{gc}(\Phi_\phi,z e^{-\phi}\,dx)$. In particular, we do not need to check whether this measure is also extremal as a canonical Gibbs measure. Therefore the assumption on $\mu$ is always nonvoid. The result on ergodicity is given in Section \ref{sec:ergodicity}.

Let us now explain the connection to the approach from \cite{DeM89} for proving the invariance principle. One of the main results in that paper states, roughly spoken, that if one has an antisymmetric integrable (plus some more properties) adapted functional of an ergodic reversible Markov process, then under diffusive scaling this functional converges to a Brownian motion scaled by a constant diffusion matrix, if has a semimartingale decomposition (the ``standard decomposition'') with a square integrable martingale and a square integrable mean forward velocity. 

If one wants to apply this to the tagged particle process from \cite{FaGr08}, the environment process plays the role of the reversible ergodic (question 3) Markov process, and the displacement of the tagged particle is the antisymmetric functional. Therefore, one needs to prove that indeed the displacement $\xi_t-\xi_0$ of the tagged particle can be described as an adapted functional on the probability space of the environment process. For this one needs to show that $\xi_t-\xi_0$ is a functional of the environment part of the coupled process (question 1) and that this environment part is distributed as the environment process (question 2).

In the case of smooth compactly supported potentials the invariance principle for the tagged particle process was studied in \cite{GP85}. Another derivation of an invariance principle, in the present setting, i.e.~with singular potentials, has been done in \cite{Os98} by means of a general theorem which is shown by abstract Dirichlet form methods. We have chosen an approach which can be seen complementary: By a deeper analysis of the particular process under consideration (the answers to question 1-3 above), which in our opinion is interesting in itself, we obtain the invariance principle as a by-product using an older and less involved framework, namely the one of \cite{DeM89}.

We summarize the core results of this article:
\begin{itemize}
\item The standard decomposition is provided for the path of the tagged particle as well as the existence of the corresponding mean forward velocity in the sense of \cite{DeM89}, see Theorem \ref{thm:martingallemma} and Lemma \ref{lem:meanforward}.
\item We represent the displacement of the tagged particle as functional of the environment, making (the use of $B^0_t$ in) equation \eqref{eqn:tag} rigorous, see Theorem \ref{thm:displacement}.
\item An essential m-dissipativity result is shown for the generator of the coupled process, which yields more control of the corresponding semigroup and shows that the coupled process from \cite{FaGr08} is the unique realization of the tagged particle dynamics given the environment process, see Theorem \ref{thm:uniqueness}.
\item Ergodicity of the environment process is shown for essentially any reference measure $\mu$ which can be considered as a pure phase of the system (i.e.~extremal elements of the set of grand canonical Gibbs measures which do not break the action of the group of uniform shifts), see Theorem \ref{thm:ergodicity}.
\item An invariance principle for the tagged particle is shown based on the results from \cite{DeM89}, see Section \ref{sec:last}. This shows that the concepts developed in that paper are applicable to martingale solutions of the above equations and in particular work for particle systems with physically realistic pair interactions.
\end{itemize}

All results apply for a general class of Ruelle measures corresponding to possibly infinite-range potentials which may have a non-integrable singularity at the origin and a nontrivial negative part and fulfill merely a weak differentiability condition on $\R^d \setminus\{0\}$. They are not restricted to small activity parameter, but instead work for any pure phase (in the sense explained above) of the system.
	
\section{Conditions on the potential and some preliminaries}\label{sec:conditions}

Throughout this article we assume that we are given a measurable function $\phi: \R^d\to\R\cup\{\infty\}$ which fulfills $\phi(x)=\phi(-x)$, $x\in\R^d$, and the following assumptions for some $p\in (d,\infty)\cap [2,\infty)$. (For some intermediate results we only need $p\in [2,\infty)$.)
\begin{enumerate}
\item[(SS)] (\emph{superstability}) There exist $A>0$ and $0\leq B<\infty$ such that for any $\eta\in\Gamma_0:=\{\gamma\in\Gamma\,|\,\sharp \gamma<\infty\}$ it holds $\Phi_\phi(\eta)\geq A\sum_{r\in\Z^d} (\sharp \eta_{Q_r})^2-B\sharp \eta$.
\item[(LR)] (\emph{lower regularity}) There exists a decreasing measurable function $\psi: [0,\infty)\to [0,\infty)$ such that $\int_0^\infty \psi(t) t^{d-1}\,dt<\infty$ and $-\phi(x)\leq \psi(\vert x\vert_2)$ holds for all $x\in\R^d$.
\item[(I)] (\emph{integrability}) It holds
$$
\int_{\R^d} \vert 1-e^{-\phi}\vert\,dx<\infty.
$$
\item[(D$\textnormal{L}^p$)] (\emph{differentiability and $L^p$}) $e^{-\phi}$ is weakly differentiable on $\R^d$ and $\phi$ is weakly differentiable on $\R^d\setminus \{0\}$. Moreover, the gradient $\nabla\phi$, considered as a $dx$-a.e.~defined function on $\R^d$, satisfies
$$
\nabla\phi\in L^1(\R^d;e^{-\phi}dx)\cap L^p(\R^d;e^{-\phi}dx).
$$
\end{enumerate}
Here $\Phi_\phi(\eta):=\sum_{\{x,y\}\subset \eta} \phi(x-y)$, $\eta\in\Gamma_0$, and for $r=(r_1,\cdots,r_d)\in\Z^d$ we set $Q_r:=\{(x_1,\cdots,x_d)\in \R^d\,|\,r_i-\frac{1}{2}\leq x_i<r_i+\frac{1}{2}\}$. $\vert \cdot\vert_2$ denotes Euclidean norm on $\R^d$. $\sharp M$ denotes the cardinality of a set $M$.

\begin{remark}\label{rem:conditions}
\begin{enumerate}
\item In \cite{FaGr08} another condition (LS) as in \cite{AKR98b} is imposed on $\phi$. We prove in the appendix that this condition can be omitted for the results derived in that paper (due to the fact that a smaller class of reference measures is considered there than in \cite{AKR98b}).
\item Condition (LR) is equivalent to the lower regularity condition from \cite{Rue70}, see (the proof of) \cite[Proposition 2.17]{KK03}. (Note that we are dealing with a pair interaction here.)
\item (SS) and (I) hold if lower regularity holds and in addition there exist $r>0$ and a decreasing function $\theta: [0,r]\to\R\cup\{\infty\}$ such that $\int_0^r \theta(s) s^{d-1}\,ds=\infty$ and $\phi(x)\geq \theta(\vert x\vert_2)$ holds whenever $x\in\R^d$, $\vert x\vert_2\leq r$, see \cite[Proposition 1.4]{Rue70}.
\item The above assumptions in particular allow the Lennard-Jones potential $\phi(x):=4\varepsilon\left(\left(\frac{s}{\vert x\vert_2}\right)^{12}-\left(\frac{s}{\vert x\vert_2}\right)^6\right)$, $x\in\R^3$, where $\varepsilon,s\in (0,\infty)$.
\end{enumerate}
\end{remark}

For any measurable function $h: \R^d\to\R$ and $\gamma\in\Gamma$ we define
$$
\langle h,\gamma\rangle:=\sum_{x\in\gamma}h(x),
$$
if this sum is absolutely convergent or $h$ is nonnegative. We sometimes use the same notation for $\R^N$-valued functions $h: \R^d\to\R^N$, $N\in\N$.

$\mathcal FC_b^\infty(\mathcal D,\Gamma)$ denotes the set of functions of the form $F=g_F(\langle f_1,\cdot\rangle,\cdots,\langle f_N,\cdot\rangle)$, where $N\in\N$, $g_F\in C_b^\infty(\R^N\to\R)$, i.e.~$g_F$ is smooth and bounded as well as its derivatives, and $f_1,\cdots f_N\in \mathcal D:=C_0^\infty(\R^d)$. For such a function $F$, the gradient $\nabla^\Gamma$ on $\Gamma$ for the natural differentiable structure on this space (see \cite{AKR98b}) is given by $$\nabla^\Gamma F(\gamma):=\left(\sum_{i=1}^N \partial_i g_F(\langle f_1,\gamma\rangle,\cdots \langle f_N,\gamma\rangle)\nabla f_i(x)\right)_{x\in\R^d}\in T_\gamma \Gamma=L^2(\R^d;\gamma),$$ 
$\gamma\in\Gamma$. Here $\gamma$ is considered as a positive Radon measure on $\R^d$ by identifying it with the sum $\sum_{x\in\gamma} \delta_x$ of point masses $\delta_x$, $x\in\gamma$. Moreover, below we use the generator $\nabla_\gamma^\Gamma$ of the group of uniform translations on $\Gamma$, which is given by $\nabla_\gamma^\Gamma F:=\sum_{i=1}^N \partial_i g_F(\langle f_1,\cdot\rangle,\cdots,\langle f_N,\cdot\rangle)\langle \nabla f_i,\cdot\rangle$ for $F$ as above. 

Embedding $\Gamma$ in the space of Radon measures on $\R^d$ as mentioned above, it can be equipped with the vague topology. The corresponding Borel $\sigma$-field is denoted by $\mathcal B(\Gamma)$. By $\mathcal G^{gc}(\Phi_\phi,z\,\sigma_\phi)$, $0<z<\infty$, we denote the set of all grand canonical Gibbs measures on $(\Gamma,\mathcal B(\Gamma))$ for $\phi$ with intensity measure $z\sigma_\phi=z\,e^{-\phi}dx$ (see e.g.~\cite[Theorem 3.12]{KK03}), and by $\mathcal G_t^{gc}(\Phi_\phi,z\,\sigma_\phi)$ we denote the set of elements of $\mathcal G^{gc}(\Phi_\phi,z\,\sigma_\phi)$ which are tempered in the sense of \cite{Rue70}. By the last part of \cite[Corollary 5.3]{Rue70} (the proof works also in the case of the intensity measure $z\sigma_\phi$, which has bounded density w.r.t.~Lebesgue measure) any $\mu\in \mathcal G_t^{gc}(\Phi_\phi,z\,\sigma_\phi)$ fulfills a Ruelle bound w.r.t.~$z\sigma_\phi$ as defined e.g.~in \cite[Definition 4.12]{KK03} (or in \cite{FaGr08}).

For $v\in C_0^\infty(\R^d)$ and $F\in\mathcal FC_b^\infty(\mathcal D,\Gamma)$ one defines $\nabla_v^\Gamma F(\gamma):=\left( v,\nabla^\Gamma F(\gamma)\right)_{T_{\gamma}\Gamma}$, $\gamma\in\Gamma$. In the appendix (which is a generalization of \cite[Theorem 3.3]{FaGr08}) it is shown that if $\phi$ fulfills (SS), (LR), (I) and (D$\textnormal{L}^1$), any $\mu\in\mathcal G_{t}^{gc}(\Phi_\phi,z\,\sigma_\phi)$ fulfills the following integration by parts formula
\begin{equation}\label{eqn:ibp1}
\int_{\Gamma} F\cdot\nabla_v^\Gamma G\,d\mu=-\int_{\Gamma}\nabla_v^\Gamma F\,\cdot G\,d\mu-\int_{\Gamma} B_v^{\phi,\mu}\,FG\,d\mu,
\end{equation}
for any $v\in C_0^\infty(\R^d)$, $F,G\in\mathcal FC_b^\infty(\mathcal D,\Gamma)$, where
$$
B_v^{\phi,\mu}(\gamma):=\langle \textnormal{div}\,v,\gamma\rangle-\sum_{x\in\gamma} (\nabla\phi(x),v(x))_{\R^d}-\sum_{\{x,y\}\subset\gamma} (\nabla\phi(x-y),v(x)-v(y))_{\R^d},
$$
$\gamma\in\Gamma$. Note that since $\mu$ fulfills an improved Ruelle bound (which can be derived from the usual Ruelle bound together with the Meyer-Montroll equation, see e.g.~\cite[Eq.~(4.29)]{AKR98b}), by \cite[Theorem 4.1]{KK02} the second sum in the previous equation converges absolutely for $\mu$-a.e.~$\gamma\in\Gamma$. Similar arguments imply the following lemma:
\begin{lemma}\label{lem:integrability}
Let $\phi$ fulfill (SS), (LR), (I), and (D$\mbox{L}^p$) for some $p\in [1,\infty)$ and let $\mu\in\mathcal G_t^{gc}(\Phi_\phi,z\,\sigma_\phi)$. Then $\langle \nabla\phi,\cdot\rangle\in L^1(\Gamma;\mu)\cap L^p(\Gamma;\mu)$ and there exists $C<\infty$ such that 
$$
\Vert \langle \nabla\phi,\cdot\rangle\Vert_{L^1(\Gamma;\mu)}+\Vert \langle \nabla\phi,\cdot\rangle\Vert_{L^p(\Gamma;\mu)}\leq C(\Vert \nabla\phi\Vert_{L^1(\R^d;\sigma_\phi)}+\Vert \nabla\phi\Vert_{L^p(\R^d;\sigma_\phi)}).
$$
\end{lemma}
\begin{proof}
It is not difficult to derive the assertion for $p\in \N$ using \cite[Theorem 4.1]{KK02} and the Ruelle bound for $\mu$. Using this, for general $p$ it can be derived by interpolating between $L^1\cap L^k$-spaces.
\end{proof}

By $\mathcal G_{ibp}^{gc}(\Phi_\phi,z\,\sigma_\phi)$ we denote the set of those $\mu\in\mathcal G_{t}^{gc}(\Phi_\phi,z\,\sigma_\phi)$ which also fulfill the additional integration by parts formula
\begin{equation}\label{eqn:ibp2}
\int_{\Gamma} F\cdot\nabla_\gamma^\Gamma G\,d\mu=-\int_{\Gamma} \nabla_\gamma^\Gamma F\cdot G\,d\mu+\int_{\Gamma}\langle \nabla\phi,\cdot\rangle FG\,d\mu, \quad F,G\in\mathcal FC_b^\infty(\mathcal D,\Gamma).
\end{equation}
Moreover, by $\mathcal G_{\theta}^{gc}(\Phi_\phi,z\,dx)$ we denote the set of translation invariant elements of the set $\mathcal G_{t}^{gc}(\Phi_\phi,z\,dx)$ of tempered grand canonical Gibbs measures for $\phi$ with intensity measure $z\,dx$. If $M$ is a convex subset of a real vector space, we denote the set of its extremal points by $\textnormal{ex}\,M$. In \cite{CK10} the following one-to-one correspondences are shown (with the exception of (iii), which is rather immediate).
\begin{lemma}\label{lem:measures}
Let $\phi$ fulfill (SS), (LR), (I), (D$\mbox{L}^1$). For $\mu_0\in \mathcal G_{t}^{gc}(\Phi_\phi,z\,dx)$ we define
$$
\Psi(\mu_0):=\frac{1}{Z_{\mu_0}}e^{-\langle \phi,\cdot\rangle}\,\mu_0,
$$
where $Z_{\mu_0}=\int_{\Gamma} e^{-\langle\phi,\cdot\rangle}\,d\mu_0$.
\begin{enumerate}
\item $\Psi: \mathcal G_t^{gc}(\Phi_\phi,z\,dx)\to \mathcal G_t^{gc}(\Phi_\phi,z\,\sigma_\phi)$ is a bijection and $\mu_0$ and $\Psi(\mu_0)$ are equivalent for all $\mu_0\in\mathcal G_t^{gc}(\Phi_\phi,z\,dx)$. Moreover, $\Psi$ maps $\mathcal G_{\theta}^{gc}(\Phi_\phi,z\,dx)$ to $\mathcal G_{ibp}^{gc}(\Phi_\phi,z\,\sigma_\phi)$.
\item If additionally $d\geq 2$ and $\nabla\phi\in L^1(\R^d\setminus B_1(0))$ then $\Psi$ is also a bijection between $\mathcal G_{\theta}^{gc}(\Phi_\phi,z\,dx)$ and $\mathcal G_{ibp}^{gc}(\Phi_\phi,z\,\sigma_\phi)$.
\item For any convex set $\mathcal M\subset G_t^{gc}(\Phi_\phi,z\,dx)$ we have $\Psi(\textnormal{ex}\,\mathcal M)=\textnormal{ex}\,\Psi(\mathcal M)$.
\end{enumerate}
\end{lemma}

In particular, since $\mathcal G_{\theta}^{gc}(\Phi_\phi,z\,dx)\neq \emptyset$ (see e.g.~\cite[Theorem 5.8]{Rue70}), we also have $\mathcal G_{ibp}^{gc}(\Phi_\phi,z\,\sigma_\phi)\neq \emptyset$.

\begin{remark}
The assertion of Lemma \ref{lem:measures}(ii) also holds in many cases if $d=1$, mainly because it is known that for many potentials there are no phase transitions in the one-dimensional setting. We remark that independently of the occurence of a phase transition Lemma \ref{lem:measures}(ii) holds in any dimension if $\phi$ fulfills the stronger assumption $\nabla\phi\in L^1(\R^d)$ (in particular if $\phi$ is bounded), as one easily derives from the proof of \cite[Lemma 3.2]{CK10}.
\end{remark}

\begin{remark}\label{rem:measuresmod}
Any $\mu\in\mathcal G_{\theta}^{gc}(\Phi_\phi,z\,dx)$ can be represented as a mixture of elements from $\textnormal{ex}\,\mathcal G_{\theta}^{gc}(\Phi_\phi, z\,dx)$, see \cite{Fol75}, \cite{Pr76}. (This is well-known, but for the reader's convenience we give some hints which hopefully simplify a complete step-by-step verification: The mentioned statement is given in Theorem 4.2 and p.~57 in \cite{Pr76}; one has to take into account that all $\Q^d$-shift invariant probability measures on $(\Gamma,\mathcal B(\Gamma))$ (which are precisely the translation invariant probability measures on $(\Gamma,\mathcal B(\Gamma))$) have a representation as integrals over extremal $\Q^d$-shift invariant, hence $\Q^d$-ergodic probability measures, which is seen as e.g.~in \cite[Theorem 14.10]{Ge88}.) In particular, $\textnormal{ex}\,\mathcal G_{\theta}^{gc}(\Phi_\phi,z\,dx)$ is nonempty. In addition any $\mu\in\textnormal{ex}\,\mathcal G_{\theta}^{gc}(\Phi_\phi,z\,dx)$ is trivial on the $\sigma$-algebra of $\mu$-a.s.~translation invariant events (see e.g.~\cite[Theorem 4.1]{Pr76}), i.e.~if $A\in\mathcal B(\Gamma)$ fulfills $\mu(\{\gamma+v\,|\,\gamma\in A\}\Delta A)=0$ for all $v\in\R^d$, then $\mu(A)\in \{0,1\}$. Here for $\gamma\in\Gamma$ we set $\gamma+v:=\{x+v\,|\,x\in\gamma\}$ and $\Delta$ denotes the symmetric set difference. The above lemma shows that these statements also hold with $\mathcal G_{\theta}^{gc}(\Phi_\phi,z\,dx)$ replaced with $\mathcal G_{ibp}^{gc}(\Phi_\phi,z\,\sigma_\phi)$, if the assumptions from part (ii) of that lemma are fulfilled.
\end{remark}

\section{The existence of the tagged particle process}\label{sec:existence}
In this section we briefly sketch the construction of the environment process and the coupled process. A detailed study of this problem can be found in \cite{FaGr08}.

\subsection{The environment process}
For $d\geq 2$ the process corresponding to \eqref{eqn:env} is realized on the configuration space $\Gamma$. For $d=1$ the process exists in general only in the larger space $\ddot\Gamma$ of multiple configurations, which can similarly as $\Gamma$ be equipped with the vague topology and the corresponding Borel $\sigma$-field $\mathcal B(\ddot \Gamma)$. For any dimension, the process is obtained by analyzing the densely defined, positive definite, symmetric bilinear form $(\mathcal E_{env}^{\Gamma,\mu},\mathcal FC_b^\infty(\mathcal D,\Gamma))$ on $L^2(\Gamma;\mu)=L^2(\ddot \Gamma;\mu)$, given by
$$
\mathcal E_{env}^{\Gamma,\mu}(F,G)=\int_{\Gamma}(\nabla^\Gamma F(\gamma),\nabla^\Gamma G(\gamma))_{T_{\gamma}\Gamma}\,d\mu(\gamma)+\int_{\Gamma} (\nabla_\gamma^\Gamma F(\gamma),\nabla_\gamma^\Gamma G(\gamma))_{\R^d}d\mu(\gamma),
$$
$F,G\in\mathcal FC_b^\infty(\mathcal D,\Gamma)$, where $\mu\in\mathcal G_{ibp}^{gc}(\Phi_\phi,z\,\sigma_\phi)$. Using \eqref{eqn:ibp1} and \eqref{eqn:ibp2}, the following result is derived in \cite{FaGr08}. (Note that although e.g.~for $p\in [1,\infty]$ the spaces $L^p(\Gamma;\mu)$ and $L^p(\ddot\Gamma;\mu)$ are the same, we sometimes emphasize the use of $\ddot\Gamma$: for example in general quasi-regularity of the Dirichlet forms used below depends on whether one uses $\ddot\Gamma$ or $\Gamma$ as state space.) Recall (cf.~Remark \ref{rem:conditions}(i)) that condition (LS) from \cite{FaGr08} can be dropped in the statements of the following theorem.
\begin{proposition}\label{lem:env1}
Suppose that $\phi$ fulfills (SS), (LR), (I) and (D$\mbox{L}^2$). Let $\mu\in\mathcal G_{ibp}^{gc}(\Phi_\phi,z\,\sigma_\phi)$. Then $(\mathcal E_{env}^{\Gamma,\mu},\mathcal FC_b^\infty(\mathcal D,\Gamma))$ is closable in $L^2(\Gamma;\mu)$ and its closure $(\mathcal E_{env}^{\Gamma,\mu},D(\mathcal E_{env}^{\Gamma,\mu}))$ is a conservative, local, quasi-regular, symmetric Dirichlet form on $L^2(\ddot \Gamma;\mu)$. Moreover, for $F,G\in\mathcal FC_b^\infty(\mathcal D,\Gamma)$ it holds
$$
\mathcal E_{env}^{\Gamma,\mu}(F,G)=\int_{\Gamma} -L_{env}^{\Gamma,\mu}F\cdot G\,d\mu,
$$
where
\begin{align*}
L_{env}^{\Gamma,\mu}F(\gamma)&=\sum_{i,j=1}^N \partial_i\partial_j g_F(\langle f_1,\gamma\rangle,\cdots,\langle f_N,\gamma\rangle)\left(\left\langle (\nabla f_i,\nabla f_j)_{\R^d},\gamma\right\rangle+\left(\langle \nabla f_i,\gamma\rangle,\langle \nabla f_j,\gamma\rangle\right)_{\R^d}\right)\\
&\quad +\sum_{j=1}^N \partial_j g_F(\langle f_1,\gamma\rangle,\cdots,\langle f_N,\gamma\rangle)\bigg( 2\langle \Delta f_j,\gamma\rangle-\left\langle\left(\nabla \phi,\nabla f_j\right)_{\R^d},\gamma\right\rangle\\
&\quad\quad\quad\quad\quad-\left(\langle\nabla\phi,\gamma\rangle,\langle \nabla f_j,\gamma\rangle\right)_{\R^d}-\sum_{\{x,y\}\subset\gamma}\left(\nabla\phi(x-y),\nabla f_j(x)-\nabla f_j(y)\right)_{\R^d}\bigg),
\end{align*}
for $\mu$-a.e.~$\gamma\in\Gamma$. The generator $(H_{env}^{\Gamma,\mu},D(H_{env}^{\Gamma,\mu}))$ of $(\mathcal E_{env}^{\Gamma,\mu},D(\mathcal E_{env}^{\Gamma,\mu}))$ is the Friedrichs extension of $(L_{env}^{\Gamma,\mu},\mathcal FC_b^\infty(\mathcal D,\Gamma))$.
\end{proposition}

Furthermore, by applying the theory of Dirichlet forms the following is shown:
\begin{theorem}
Under the assumptions of Proposition \ref{lem:env1} the following holds:
\begin{enumerate}
\item There exists a (up to $\mu$-equivalence unique, see \cite[Section IV.6]{MaRo92}) conservative diffusion process
$$
\mathbf M_{env}^{\Gamma,\mu}=\left(\mathbf \Omega^{env},\mathbf F^{env},(\mathbf F_t^{env})_{t\geq 0},(\mathbf X_t^{env})_{t\geq 0},(\mathbf P_{\gamma}^{env})_{\gamma\in\ddot\Gamma}\right)
$$
on $\ddot \Gamma$ which is properly associated with $(\mathcal E_{env}^{\Gamma,\mu},D(\mathcal E_{env}^{\Gamma,\mu}))$, i.e.~for all ($\mu$-versions of) $F\in L^2(\ddot \Gamma;\mu)$ and all $t>0$ the function 
$$
\gamma\mapsto p_t^{env}F(\gamma):=\mathbf E_{\gamma}^{env}[F(\mathbf X_t^{env})],\quad \gamma\in\ddot\Gamma,
$$
is an $\mathcal E_{env}^{\Gamma,\mu}$-quasi continuous $\mu$-version of $T_{t,2}^{env}F:=\exp(t H_{env}^{\Gamma,\mu})F$. (Here $\mathbf E_{\gamma}^{env}$ denotes the expectation corresponding to $\mathbf P_{\gamma}^{env}$; below we also use analogous notations.) In particular, $\mathbf M_{env}^{\Gamma,\mu}$ is $\mu$-symmetrie, i.e.~$\int_{\Gamma} p_t^{env}F\cdot G\,d\mu=\int_{\Gamma}F\cdot p^{env}_t G\,d\mu$ for all nonnegative $\mathcal B(\ddot\Gamma)$-measurable $F,G: \ddot\Gamma\to\R$.
\item $\mathbf M_{env}^{\Gamma,\mu}$ has invariant measure $\mu$ and solves the martingale problem for $(H_{env}^{\Gamma,\mu},D(H_{env}^{\Gamma,\mu}))$ in the following sense: For all $G\in D(H_{env}^{\Gamma,\mu})$, thus in particular for any $G\in\mathcal FC_b^\infty(\mathcal D,\Gamma)$, the process defined by 
$$
\tilde G(\mathbf X_t^{env})-\tilde G(\mathbf X_{0}^{env})-\int_0^t H_{env}^{\Gamma,\mu} G(\mathbf X_s^{env})\,ds,\quad t\geq 0,
$$
is an $(\mathbf F_t^{env})_{t\geq 0}$-martingale under $\mathbf P_\gamma^{env}$ (hence starting at $\gamma$) for $\mathcal E_{env}^{\Gamma,\mu}$-quasi-all $\gamma\in\ddot\Gamma$. Here $\tilde G$ denotes an $\mathcal E_{env}^{\Gamma,\mu}$-quasi-continuous $\mu$-version of $G$.
\end{enumerate}
\end{theorem}

\subsection{The coupled process}

For constructing and analyzing the Markov process on $\R^d\times\Gamma$ representing the coupled motion of the tagged particle and its environment, we often consider test functions $\mathfrak F\in C_0^\infty(\R^d)\otimes \mathcal FC_b^\infty(\mathcal D,\Gamma)$, i.e.~functions of the form $\mathfrak F=\sum_{i=1}^k f_i\otimes F_i$, $f_1,\cdots,f_k\in C_0^\infty(\R^d)$, $F_1,\cdots F_k\in\mathcal FC_b^\infty(\mathcal D,\Gamma)$, $k\in\N$, where for $f\in C_0^\infty(\R^d)$ and $F\in\mathcal FC_b^\infty(\mathcal D,\Gamma)$ we set $f\otimes F(\xi,\gamma):=f(\xi)F(\gamma)$, $(\xi,\gamma)\in\R^d\times\Gamma$. Some of the linear operators on $C_0^\infty(\R^d)\otimes \mathcal FC_b^\infty(\mathcal D,\Gamma)$ we consider below are the operators $(\nabla_\gamma^\Gamma-\nabla_\xi)$, $\nabla^\Gamma$, $(\nabla_\gamma^\Gamma,\nabla_\xi)_{\R^d}$, given by
\begin{align*}
(\nabla_\gamma^\Gamma-\nabla_\xi)\mathfrak F(\xi,\gamma)&:=\sum_{i=1}^k \left(f_i(\xi)\nabla_\gamma^\Gamma F_i-F_i(\gamma)\nabla f_i(\xi) \right)\\
\nabla^\Gamma \mathfrak F(\xi,\gamma)&:=\sum_{i=1}^k f_i(\xi)\nabla^\Gamma F_i(\gamma)\\
(\nabla_\gamma^\Gamma,\nabla_\xi)_{\R^d}\mathfrak F(\xi,\gamma)&:=\sum_{i=1}^k (\nabla_\gamma^\Gamma F_i(\gamma),\nabla f_i(\xi))_{\R^d}
\end{align*}
for $(\xi,\gamma)\in\R^d\times \Gamma$ and $\mathfrak F=\sum_{i=1}^k f_i\otimes F_i$ as above. As a general rule, the operators $L_{env}^{\Gamma,\mu}$, $\nabla_\gamma^\Gamma$ and $\nabla^\Gamma$ are supposed to act only on the $\gamma$-dependent parts of an $\mathfrak F\in C_0^\infty(\R^d)\otimes \mathcal FC_b^\infty(\mathcal D,\Gamma)$, while $\nabla_\xi$ and $\Delta_\xi$ only act on the $\xi$-dependent parts. We also use similar notations in analogous situations (i.e.~for tensor products of other spaces of functions on $\R^d$ and $\Gamma$). Moreover, we sometimes consider a function $F: \Gamma\to\R$ as a function $F: \R^d\times\Gamma\to\R$ by setting $F(\xi,\gamma):=F(\gamma)$ for $(\xi,\gamma)\in\R^d\times\Gamma$. 

The coupled process is obtained by analyzing the densely defined, positive definite, symmetric bilinear form $(\mathcal E_{coup}^{\R^d\times \Gamma,\hat\mu},C_0^\infty(\R^d)\otimes \mathcal FC_b^\infty(\mathcal D,\Gamma))$ given by
\begin{multline*}
\mathcal E_{coup}^{\R^d\times\Gamma,\hat\mu}(\mathfrak F,\mathfrak G):=\int_{\R^d\times\Gamma} \left(\nabla^\Gamma \mathfrak F(\xi,\gamma),\nabla^\Gamma \mathfrak G(\xi,\gamma)\right)_{T_\gamma \Gamma}\,d\hat\mu(\xi,\gamma)\\
+\int_{\R^d\times\Gamma} \left( (\nabla_\gamma^\Gamma-\nabla_\xi)\mathfrak F(\xi,\gamma),(\nabla_\gamma^\Gamma-\nabla_\xi)\mathfrak G(\xi,\gamma)\right)_{\R^d}d\hat\mu(\xi,\gamma),
\end{multline*}
$\mathfrak F,\mathfrak G\in C_0^\infty(\R^d)\otimes \mathcal FC_b^\infty(\mathcal D,\Gamma)$, where $\hat\mu:=d\xi\otimes\mu$ and $\mu\in\mathcal G_{ibp}^{gc}(\Phi_\phi,z\,\sigma_\phi)$.

\begin{proposition}\label{lem:coup1}
Suppose that $\phi$ fulfills (SS), (LR), (I) and (D$\mbox{L}^p$) for $p\geq \max\{2,d\}$. Furthermore, let $\mu$ and $\hat\mu$ be as above. Then $(\mathcal E_{coup}^{\R^d\times\Gamma,\hat\mu},D(\mathcal E_{coup}^{\R^d\times\Gamma,\hat\mu}))$ is closable in $L^2(\R^d\times\Gamma;\hat\mu)$ and its closure $(\mathcal E_{coup}^{\R^d\times\Gamma,\hat\mu},D(\mathcal E_{coup}^{\R^d\times\Gamma,\hat\mu}))$ is a conservative, local, quasi-regular, symmetric Dirichlet form on $L^2(\R^d\times\ddot\Gamma;\hat\mu)$. Moreover, for $\mathfrak F,\mathfrak G\in C_0^\infty(\R^d)\otimes \mathcal FC_b^\infty(\mathcal D,\Gamma)$ it holds
$$
\mathcal E_{coup}^{\R^d\times\Gamma,\hat\mu}(\mathfrak F,\mathfrak G)=\int_{\R^d\times\Gamma} -L_{coup}^{\R^d\times\Gamma,\hat\mu}\mathfrak F\cdot\mathfrak G\,d\hat\mu,
$$
where
$$
L_{coup}^{\R^d\times\Gamma,\hat\mu}=L_{env}-2(\nabla_\gamma^\Gamma,\nabla_\xi)_{\R^d}+(\langle \nabla\phi,\cdot\rangle,\nabla_\xi)_{\R^d}+\Delta_\xi.
$$
The generator $\left(H_{coup}^{\R^d\times\Gamma,\hat\mu},D(H_{coup}^{\R^d\times\Gamma,\hat\mu})\right)$ of $\left(\mathcal E_{coup}^{\R^d\times\Gamma,\hat\mu},D(\mathcal E_{coup}^{\R^d\times\Gamma,\hat\mu})\right)$ is the Friedrichs extension of $\left(L_{coup}^{\R^d\times\Gamma},C_0^\infty(\R^d)\otimes \mathcal FC_b^\infty(\mathcal D,\Gamma)\right)$. 
\end{proposition}

Applying the theory of Dirichlet forms the following is shown in \cite{FaGr08}:
\begin{theorem}\label{lem:coup2}
Under the assumptions of Proposition \ref{lem:coup1} the following holds:
\begin{enumerate}
\item There exists a (up to $\hat\mu$-equivalence unique) conservative diffusion process
$$
\mathbf M_{coup}^{\R^d\times\Gamma,\hat\mu}=(\mathbf \Omega^{coup},\mathbf F^{coup},(\mathbf F_t^{coup})_{t\geq 0},(\mathbf X_t^{coup})_{t\geq 0},(\mathbf P_{(\xi,\gamma)}^{coup})_{(\xi,\gamma)\in\R^d\times\ddot\Gamma})
$$
on $\R^d\times\ddot\Gamma$ which is properly associated with $\left(\mathcal E_{coup}^{\R^d\times\Gamma,\hat\mu},D(\mathcal E_{coup}^{\R^d\times\Gamma,\hat\mu})\right)$. In particular, $\mathbf M_{coup}^{\R^d\times\Gamma,\hat\mu}$ is $\hat\mu$-symmetric.
\item $\mathbf M_{coup}^{\R^d\times\Gamma,\hat\mu}$ has $\hat\mu$ as invariant measure and solves the martingale problem for the operator $\left(H_{coup}^{\R^d\times\Gamma,\hat\mu},D(H_{coup}^{\R^d\times\Gamma,\hat\mu})\right)$ in the following sense: For all $\mathfrak G\in D(H_{coup}^{\R^d\times\Gamma,\hat\mu})$, in particular for all $\mathfrak G\in C_0^\infty(\R^d)\otimes \mathcal FC_b^\infty(\mathcal D,\Gamma)$, the process defined by
$$
\widetilde {\mathfrak G}(\mathbf X_t^{coup})-\widetilde {\mathfrak G}(\mathbf X_0^{coup})-\int_0^t H_{coup}^{\R^d\times\Gamma,\hat\mu}\mathfrak G(\mathbf X_s^{coup})\,ds,\quad t\geq 0,
$$
is an $(\mathbf F_t^{coup})_{t\geq 0}$-martingale under $\mathbf P_{(\xi,\gamma)}^{coup}$ for $\mathcal E_{coup}^{\R^d\times\Gamma,\hat\mu}$-quasi-all $(\xi,\gamma)\in \R^d\times\ddot\Gamma$. Here $\widetilde{\mathfrak G}$ denotes an $\mathcal E_{coup}^{\R^d\times\Gamma,\hat\mu}$-quasi-continuous $\hat\mu$-version of $\mathfrak G$.
\end{enumerate}
\end{theorem}

Below we use the notation $T_{t,2}^{coup}:=\exp(tH_{coup}^{\R^d\times\Gamma,\hat\mu})$, $t\geq 0$.

\begin{remark}
It is important to note that $\mathbf M^{env}$ (resp.~$\mathbf M^{coup}$) is a solution of \eqref{eqn:env} (resp.~\eqref{eqn:tag} and \eqref{eqn:env}) in the sense that the martingale problem for the corresponding generators is solved as stated in the previous lemmas. There would be at least some work to do in order to show that this gives rise to solutions in a stronger sense (e.g.~proving that one can enumerate all the particles from the environment and construct their trajectories, which should not explode). The martingale problem (or the associatedness of the processes with the respective Dirichlet form) is therefore the only possible starting point for deriving results on these processes.
\end{remark}

Sometimes we need the statement on the martingale problem also for some functions which are ``locally'' in $C_0^\infty(\R^d)\otimes \mathcal FC_b^\infty(\mathcal D,\Gamma)$. Therefore we prove the following lemma. We denote by $\mathbf X^{coup,1}_t$ and $\mathbf X^{coup,2}_t$ the two components of $\mathbf X^{coup}_t$ for $t\geq 0$.
\begin{theorem}\label{thm:martingallemma}
\begin{enumerate}
\item Let $F\in\mathcal FC_b^\infty(\mathcal D,\Gamma)$ and let $c\in \R^d$ (e.g.~some standard unit vector or $0$). Let $\mathfrak F(\xi,\gamma):=F(\gamma)+(c,\xi)_{\R^d}$, $(\xi,\gamma)\in\R^d\times\Gamma$ and set $L^{\R^d \times\Gamma,\hat\mu}_{coup}\mathfrak F(\xi,\gamma):=L_{env}^{\Gamma,\mu}F(\gamma)+\langle (c,\nabla\phi)_{\R^d},\gamma\rangle$. Then for $\mathcal E_{coup}^{\R^d\times\Gamma,\hat\mu}$-quasi every $(\xi,\gamma)\in\R^d\times\Gamma$
$$
M_t^{\mathfrak F}:=\mathfrak F(\mathbf X^{coup}_t)-\mathfrak F(\mathbf X^{coup}_0)-\int_0^t L^{\R^d \times\Gamma,\hat\mu}_{coup}\mathfrak F(\mathbf X^{coup}_s)\,ds,\quad t\geq 0,
$$
defines a continuous local $(\mathbf F_t^{coup})_{t\geq 0}$-martingale under $\mathbf P^{coup}_{(\xi,\gamma)}$ with quadratic variation process given by
\begin{equation}\label{eqn:qvariation}
\langle M^{\mathfrak F}\rangle_t=2\int_0^t (\nabla^\Gamma F(\mathbf X^{coup,2}_s),\nabla^\Gamma F(\mathbf X^{coup,2}_s))_{T_{\mathbf X^{coup,2}_s}\Gamma}+\left\vert \nabla^\Gamma F(\mathbf X^{coup,2}_s)-c\right\vert^2 \,ds.
\end{equation}
\item For $\mathcal E_{coup}^{\R^d\times\Gamma,\hat\mu}$-quasi every $(\xi,\gamma)\in\R^d\times\Gamma$ the process 
$$
\left(\mathbf X^{coup,1}_t-\mathbf X^{coup,1}_0-\int_0^t \langle\nabla\phi,\mathbf X^{coup,2}_s\rangle ds\right)_{t\geq 0}
$$ 
is distributed under $\mathbf P^{coup}_{(\xi,\gamma)}$ like $\sqrt{2}$ times a $d$-dimensional Brownian motion starting in $0$.
\end{enumerate}
\end{theorem}
\begin{proof}
For $k\in\N$ choose a function $\chi_k\in C_0^\infty(\R^d)$ such that $1_{[-k,k]^d}\leq \chi_k\leq 1_{[-k-1,k+1]^d}$. Considering $\chi_k$ as a function on $\R^d\times\Gamma$ which is constant in the second argument, we obtain functions $\mathfrak F_k:=\chi_k \mathfrak F\in C_0^\infty(\R^d)\times \mathcal FC_b^\infty(\mathcal D,\Gamma)$. By Lemma \ref{lem:coup2}(ii) and the fact that the set of $\mathcal E_{coup}^{\R^d\times\Gamma,\hat\mu}$-exceptional sets is closed w.r.t.~countable unions, there is an $\mathcal E_{coup}^{\R^d\times\Gamma,\hat\mu}$-exceptional set $N\subset \R^d\times\Gamma$ such that for all $(\xi,\gamma)\in (\R^d\times\Gamma)\setminus N$ all $(M^{\mathfrak F_k}_t)_{t\geq 0}$, $k\in\N$, are $(\mathbf F_t^{coup})_{t\geq 0}$ martingales under $\mathbf P^{coup}_{(\xi,\gamma)}$ with quadratic variation processes given as in \eqref{eqn:qvariation}, with $\mathfrak F$ replaced by $\mathfrak F_k$. The latter statement follows using \cite[Theorems 5.1.3 and 5.2.3]{FOT94} as in \cite[Example 5.2.1]{FOT94}.

Since $\mathbf M^{coup}$ is conservative, we may in addition assume that $\mathbf P^{coup}_{(\xi,\gamma)}$ is conservative for all $(\xi,\gamma)\in\R^d\times\Gamma\setminus N$. From this we obtain using the $(\mathbf F_t^{coup})_{t\geq 0}$-stopping times $\tau_k:=\inf\{t\geq 0\,|\,\mathbf X^{coup,1}\notin [-k,k]^d\}$, $k\in\N$, that $(M_{t\wedge \tau_k}^{\mathfrak F})_{t\geq 0}=(M_{t\wedge \tau_k}^{\mathfrak F_k})_{t\geq 0}$ is for each $k\in\N$ an $(\mathbf F_t^{coup})_{t\geq 0}$-martingale under $\mathbf P^{coup}_{(\xi,\gamma)}$ for all $(\xi,\gamma)\in (\R^d\times\Gamma)\setminus N$ with quadratic variation process given by $(\langle M^{\mathfrak F}\rangle_{t\wedge \tau_k})_{t\geq 0}=(\langle M^{\mathfrak F_k}\rangle_{t\wedge \tau_k})_{t\geq 0}$. (The $\tau_k$ are stopping times by right-continuity of the filtration $(\mathbf F_t^{coup})_{t\geq 0}$.) Since by conservativity we have $\tau_k\to \infty$ as $k\to\infty$ $\mathbf P^{coup}_{(\xi,\gamma)}$-a.s., it follows that $(M_t^{\mathfrak F})_{t\geq 0}$ is a local martingale with quadratic variation process given as in \eqref{eqn:qvariation}. The second assertion is now a consequence of L\'{e}vy's theorem (see e.g. \cite{EK86}).
\end{proof}

\begin{remark}\label{rem:martingale}
By the Burkholder-Davis-Gundy inequality and \cite[Theorem (2.5)]{Dur96} we can drop the ``local'' in the statement of the above lemma if we can prove that for $\mathcal E_{coup}^{\R^d\times\Gamma,\hat\mu}$-quasi every $(\xi,\gamma)\in \R^d\times\Gamma$ and all $T\in\N$ it holds
$$
\mathbf E^{coup}_{(\xi,\gamma)}\langle M^{\mathfrak F}\rangle_T<\infty,
$$
which is e.g.~true if
\begin{equation}\label{eqn:Resolventexp}
\mathbf E^{coup}_{(\xi,\gamma)} \int_0^\infty e^{-s}Z(\mathbf X^{coup}_s) \,ds<\infty,
\end{equation}
where $Z(\xi',\gamma'):=(\nabla^\Gamma F(\gamma'),\nabla^\Gamma F(\gamma'))_{T_{\gamma'}\Gamma}+\left\vert \nabla_\gamma^\Gamma F(\gamma')-c\right\vert^2$, $(\xi',\gamma')\in\R^d\times\Gamma$. We believe that this is true, but do not further investigate this question, since below we only consider mixtures of the laws $\mathbf P_{(\xi,\gamma)}^{coup}$, for which \eqref{eqn:Resolventexp} is clear.
\end{remark}

\section{One-particle uniqueness and the environment part of $(T_{t,2}^{coup})_{t\geq 0}$.}\label{sec:uniqueness}

Let $\phi$ fulfill the assumptions from the beginning of Section \ref{sec:conditions}. Let $\mu\in\mathcal G_{ibp}^{gc}(\Phi_\phi,z\,\sigma_\phi)$. By symmetry of the sub-Markovian semigroup $(T_{t,2}^{coup})_{t\geq 0}$ and by the Beurling-Deny theorem (see \cite[Proposition 1.8]{LS96}) there exists for each $q\in [1,\infty]$ a contraction semigroup $(T_{t,q}^{coup})_{t\geq 0}$ on $L^q(\R^d\times\Gamma;\hat\mu)=:L^q$ extending the restriction of $(T_{t,2}^{coup})_{t\geq 0}$ to $L^1\cap L^\infty$, such that for each $q<\infty$ the semigroup $(T_{t,q}^{coup})_{t\geq 0}$ is strongly continuous and $(T_{t,\frac{q}{q-1}}^{coup})_{t\geq 0}$ is the adjoint semigroup of $(T_{t,q}^{coup})_{t\geq 0}$. For convenience, we sometimes drop the index $q$ and simply write $(T_{t}^{coup})_{t\geq 0}$. For identifying the second component of the coupled process as the environment process (see also Lemma \ref{lem:identification1} below), the main step is to prove the identity
\begin{equation}\label{eqn:decomposition}
T_{t,\infty}^{coup}(1\otimes F)=1\otimes T_{t,2}^{env}F,\quad F\in L^\infty(\Gamma;\mu).
\end{equation}
One proof for this is given in \cite[Proof of Lemma 2.3]{Os98}. However, as explained in the introduction, we give a different proof based on showing a result on essential m-dissipativity of the generator $H_{coup}^{\R^d\times\Gamma,d\xi\otimes \mu}$ on a (large but) suitable domain. This gives us not only sufficient know\-ledge of $(T^{coup}_t)_{t\geq 0}$ for obtaining \eqref{eqn:decomposition}, but even shows that, given the environment process, there exists only one possible coupled process (see Remark \ref{rem:uniqueness} below).

\subsection{A one-particle uniqueness result}\label{sub:uniqueness}

The aim of this section is to prove essential m-dissipativity of $H_{coup}^{\R^d\times\Gamma,\hat\mu}$ in some $L^q$, $q\in [1,2]$, on a rather large set of functions on $\R^d\times \Gamma$, which is nevertheless small enough for our purposes. Since these functions are of a very simple type only in the first component (and since knowledge of a generator core uniquely determines the semigroup), this result is named one-particle uniqueness. We need some preparations. By $\mathcal S(\R^d)$ we denote the Schwartz space of rapidly decreasing smooth functions on $\R^d$ and by $D(H_{env}^{\Gamma,\mu})_b$ we denote the set of bounded functions from $D(H_{env}^{\Gamma,\mu})$. Any $F\in D(H_{env}^{\Gamma,\mu})$ is contained in $D(\mathcal E_{env}^{\Gamma,\mu})$. In particular, for such $F$ the object $\nabla_\gamma^\Gamma F\in L^2(\Gamma\to\R;\mu)$ is reasonably defined as the limit of $(\nabla_\gamma^\Gamma F_n)_{n\in\N}$ for a sequence $(F_n)_{n\in\N}\subset \mathcal FC_b^\infty(\mathcal D,\Gamma)$ approximating $F$ in $D(\mathcal E_{env}^{\Gamma,\mu})$ w.r.t.~the norm $\Vert\cdot\Vert_{D(\mathcal E_{env}^{\Gamma,\mu})}:=\sqrt{(\cdot,\cdot)_{L^2(\Gamma;\mu)}+\mathcal E_{env}^{\Gamma,\mu}(\cdot,\cdot)}$.

\begin{lemma}\label{lem:preparationdomain}
Let $f\in\mathcal S(\R^d)$ and $F\in D(H_{env}^{\Gamma,\mu})_b$. Then $\mathfrak F:=f\otimes F\in D(H_{coup}^{\R^d\times\Gamma,\hat\mu})$ and
$$
H_{coup}^{\R^d\times\Gamma,\hat\mu}\mathfrak F= H_{env}^{\Gamma,\mu}\mathfrak F-2(\nabla_\gamma^\Gamma,\nabla_\xi)_{\R^d}\mathfrak F+(\langle \nabla\phi,\cdot\rangle,\nabla_\xi)_{\R^d}\mathfrak F+\Delta_\xi \mathfrak F.
$$
\end{lemma}
\begin{proof}
Let at first $F\in\mathcal FC_b^\infty(\mathcal D,\Gamma)$ and $f\in C_0^\infty(\R^d)$. Then for any $g\in C_0^\infty(\R^d)$ and $G\in \mathcal FC_b^\infty(\mathcal D,\Gamma)$ it holds
\begin{eqnarray}\label{eqn:hilfslemmaeq}
\lefteqn{\mathcal E_{coup}^{\R^d\times\Gamma,\hat\mu}(f\otimes F,g\otimes G)}\nonumber\\
& &=\mathcal E_{env}^{\Gamma,\mu}(F,G)(f,g)_{L^2(\R^d)}+\int_{\Gamma}\int_{\R^d}F(\gamma) G(\gamma) (\nabla_\xi f(\xi),\nabla_\xi g(\xi))_{\R^d}\,d\xi \,d\mu(\gamma)\nonumber\\
& &\quad\quad -\int_{\Gamma}\int_{\R^d} f(\xi) G(\gamma)(\nabla_\gamma^\Gamma F(\gamma),\nabla_\xi g(\xi))_{\R^d}+g(\xi)F(\gamma)(\nabla_\gamma^\Gamma G(\gamma),\nabla_\xi f(\xi))_{\R^d}\,d\xi\,d\mu(\gamma)\nonumber\\
& &=\mathcal E_{env}^{\Gamma,\mu}(F,G)(f,g)_{L^2(\R^d)}+\left(-\Delta_\xi f\otimes F+2(\nabla_\gamma^\Gamma F,\nabla_\xi f)_{\R^d},g\otimes G\right)_{L^2}\nonumber\\
& &\quad\quad -\int_\Gamma\int_{\R^d} (\langle \nabla\phi,\gamma\rangle,\nabla_\xi f(\xi))_{\R^d} F(\gamma)g(\xi)G(\gamma)\,d\xi\,d\mu(\gamma),
\end{eqnarray}
see also \cite[Proof of Theorem 4.15]{FaGr08}. This immediately extends to $f,g\in\mathcal S(\R^d)$. Moreover, it also extends to $F\in D(\mathcal E_{env}^{\Gamma,\mu})$, as one sees by approximating $F$ by $(F_n)_{n\in\N}\subset \mathcal FC_b^\infty(\mathcal D,\Gamma)$ as mentioned above: Convergence of the right-hand side is shown using the considerations preceding this lemma. To see convergence of the left-hand side we note that for such a sequence we have $\sup_{n\in\N}\mathcal E_{coup}^{\R^d\times\Gamma,\hat\mu}(f\otimes F_n,f\otimes F_n)<\infty$, hence by \cite[Lemma I.2.12]{MaRo92} it follows $f\otimes F\in D(\mathcal E_{coup}^{\R^d\times\Gamma,\hat\mu})$ and $\mathcal E_{coup}^{\R^d\times\Gamma,\hat\mu}(f\otimes F_n,g\otimes G)\to \mathcal E_{coup}^{\R^d\times\Gamma,\hat\mu}(f\otimes F,g\otimes G)$ as $n\to\infty$.

For any bounded $F\in D(\mathcal E_{env}^{\Gamma,\mu})$ we find that the last summand on the right-hand side of \eqref{eqn:hilfslemmaeq} can be rewritten as $L^2$-inner product $((\nabla_\xi f)\otimes (\langle\nabla\phi,\cdot\rangle\,F),g\otimes G)_{L^2}$, and if in addition $F\in D(H_{env}^{\Gamma,\mu})$, the first summand equals $-(f\otimes H_{env}^{\Gamma,\mu}F,g\otimes G)_{L^2}$. The assertion follows from \cite[Proposition I.2.16]{MaRo92} and denseness of the linear span of functions $g\otimes G$ as above in $D(\mathcal E_{coup}^{\R^d\times\Gamma,\hat\mu})$ w.r.t.~$\Vert\cdot\Vert_{D(\mathcal E_{coup}^{\R^d \times\Gamma,\hat\mu})}$.
\end{proof}

For any $H\in L^2(\Gamma\to\R^d;\mu)$ we define a linear operator $(L_H,\tilde D)$ on $L^q(\R^d\times\Gamma;\hat\mu)$, $q\in [1,2]$, by 
$$
L_H\mathfrak F:=H_{env}^{\Gamma,\mu}\mathfrak F+\Delta_\xi\mathfrak F-2(\nabla_\gamma^\Gamma,\nabla_\xi)_{\R^d}\mathfrak F+(H,\nabla_\xi)_{\R^d}\mathfrak F,
$$
$\mathfrak F\in \tilde D:=\mathcal S(\R^d)\otimes \mathcal FC_b^\infty(\mathcal D,\Gamma)$.

$H_{env}^{\Gamma,\mu}$, considered as an operator acting on $\tilde D$, is dissipative in any $L^q$, $q\in [1,2]$, since an extension of it generates in $L^2$ the symmetric sub-Markovian (and thus $L^p$-contractive) strongly continuous contraction semigroup $(I\otimes T_{t,2}^{env})_{t\geq 0}$. Dissipativity in any $L^q$, $q\in [1,2]$, of $(H_{coup}^{\R^d\times\Gamma,\hat\mu},\tilde D)=(L_{H_{\phi}},\tilde D)$, where $H_{\phi}:=\langle \nabla\phi,\cdot\rangle$, follows by an analogous argument.

\begin{lemma}
Let $H\in L^2(\Gamma\to\R^d;\mu)$. Then $(L_H,\tilde D)$ is dissipative in any $L^q$, $q\in [1,2]$.
\end{lemma} 
\begin{proof}
Let $\mathfrak F\in \tilde D$. For $\varepsilon>0$ let $\varphi_\varepsilon\in C^1(\R)$ be such that $\varphi_{\varepsilon}(t)=0$ for $t\leq 0$, $\varphi_\varepsilon'$ is nondecreasing and $\varphi_\varepsilon'(t)=t^{q-1}$ for $t\geq \varepsilon$. Then the funtion $\varphi_\varepsilon\circ\mathfrak F$ is differentiable w.r.t.~$\xi$ and decreases as well as its derivative quickly at $\xi=\infty$. By integration by parts in the $\xi$-directions we obtain
$$
\int_{\R^d\times \Gamma}\left((H-H_\phi),\nabla_\xi(\varphi_\varepsilon\circ\mathfrak F)\right)_{\R^d}d\hat\mu=0.
$$
Using the chain rule and letting $\varepsilon\to 0$ we obtain by Lebesgue's dominated convergence theorem $\int_{\R^d\times\Gamma} (L_H-H_{coup}^{\R^d\times\Gamma,\hat\mu})\mathfrak F\cdot (\mathfrak F^+)^{q-1}\,d\hat\mu=0$, where $\mathfrak F^+$ denotes the positive part of $\mathfrak F$. Adding the same equality with $\mathfrak F$ replaced by $-\mathfrak F$, we obtain $\int_{\R^d\times\Gamma} (L_H-H_{coup}^{\R^d\times\Gamma,\hat\mu})\mathfrak F\cdot\vert \mathfrak F\vert^{q-1}\textnormal{sign}(\mathfrak F)\,d\hat\mu=0$. This implies the assertion, since $H_{coup}^{\R^d\times\Gamma,\hat\mu}|_{\tilde D}$ is the restriction of an m-dissipative operator in $L^q$.
\end{proof}

We consider the operators $L_H$ on the domains $\tilde D$, $D:=C_0^\infty(\R^d)\otimes D(H_{env}^{\Gamma,\mu})_b$ and $\hat D:=\mathcal F^{-1}(C_0^\infty(\R^d))\otimes D(H_{env}^{\Gamma,\mu})_b$. Here $\mathcal F: L^2(\R^d)\to L^2(\R^d)$ denotes the Fourier transform. $L_H$ is dissipative on all these domains in all $L^q$, $q\in [1,2]$, and hence closable. Moreover, since $C_0^\infty(\R^d)$ and $\mathcal F^{-1}(C_0^\infty(\R^d))$ are dense in $\mathcal S(\R^d)$ w.r.t.~its usual Frechet topology and the latter is continuously and densely embedded in any Sobolev space $H^{m,p}(\R^d)$, $m\in\N_0$, $p\in [1,\infty)$, we obtain the following lemma (which is mentioned here rather for completeness and for convenience of the reader).
\begin{lemma}
Let $H\in L^2(\Gamma\to\R^d;\mu)$. Then the closures of $(L_H,D)$, $(L_H,\tilde D)$ and $(L_H,\hat D)$ (exist and) coincide in all $L^q$, $q\in [1,2]$.
\end{lemma}

For a moment let us restrict to the $L^2$-setting. Below we sometimes switch to the complexified function spaces. By $\mathcal F_\xi: L^2(\R^d\times\Gamma;\hat\mu)\to L^2(\R^d\times\Gamma;\hat\mu)$ we denote the Fourier transform in the first argument, given as the unique unitary operator fulfilling $\mathcal F_\xi(f\otimes F)=(\mathcal F f)\otimes F$, $F\in L^2(\Gamma;\mu)$, $f\in L^2(\R^d)$. For $k\in\N$ let $P_k: L^2\to L^2$ be the orthogonal projection given by multiplication with the indicator function $1_{\overline{B_k(0)\times \Gamma}}$. We define spaces $L^2_k:=\mathcal F_\xi^{-1}(\textnormal{Range}(P_k))$, $k\in\N$, of $L^2$-functions with compact band in the $\xi$-directions. The set $\hat D_k:=\mathcal F^{-1}(C_0^\infty(B_k(0)))\otimes D(H_{env}^{\Gamma,\mu})_b$ is dense in $L_k^2$.

Note that for $H\in L^2(\Gamma\to\R^d;\mu)$ the unitary operator $\mathcal F_\xi$ transforms $(L_H,\tilde D)$ into the operator $(\hat L_H,\tilde D)$, given by
\begin{equation}\label{eqn:hatLH}
\hat L_H\mathfrak F(\xi,\gamma):=H_{env}^{\Gamma,\mu}\,\mathfrak F(\xi,\gamma)-\vert \xi\vert_2^2\mathfrak F(\xi,\gamma)+2i(\xi,\nabla_\gamma^\Gamma\mathfrak F(\xi,\gamma))_{\R^d}-i(H(\gamma),\xi)_{\R^d}\mathfrak F(\xi,\gamma),
\end{equation}
$\mathfrak F\in \tilde D$, $(\xi,\gamma)\in \R^d\times\Gamma$. This shows that for $H\in L^2(\Gamma\to\R^d;\mu)$ the operator $(L_H,\hat D_k)$ is well-defined as an operator in $L_k^2$, $k\in\N$.

\begin{lemma}\label{lem:uniquenessprep}
Let $H\in L^\infty(\Gamma\to\R^d;\mu)$. Then for every $k\in\N$ the operator $(L_H,\hat D_k)$ is essentially m-dissipative in $L_k^2$. It follows that $(L_H,\hat D)$ (hence also $(L_H,D)$) is essentially m-dissipative in $L^2(\R^d\times\Gamma;\hat\mu)$.
\end{lemma}
(Recall that essential m-dissipativity of an operator $(A,D(A))$ on a Banach space $X$ means that $A$ is dissipative and $\textnormal{Range}(I-A)$ is dense in $X$. The well-known Lumer-Phillips theorem implies that in this case the closure of $A$ is the unique extension of $A$ which generates a strongly continuous contraction semigroup.)
\begin{proof}
The second assertion follows from the first, since the set $\bigcup_{k\in\N}L_k^2$ is dense in $L^2$. To prove the first assertion, we note that the operator $(H_{env}^{\Gamma,\mu},\hat D_k)$ is essentially self-adjoint as an operator in $L_k^2$ and hence also essentially m-dissipative in this space. We prove that the (antisymmetric, hence dissipative) operator $(L_H-H_{env}^{\Gamma,\mu},\hat D_k)$ is Kato bounded by $H_{env}^{\Gamma,\mu}$ with bound $0$. The assertion then follows from standard perturbation theory. The easiest way to show the Kato boundedness is to consider the images of $L_H$ and $H_{env}^{\Gamma,\mu}$ w.r.t.~$\mathcal F_\xi$, which are defined on $D_k:=C_0^\infty(B_k(0))\otimes D(H_{env}^{\Gamma,\mu})_b$. Using \eqref{eqn:hatLH} we find that
$$
\Vert (\hat L_H-H_{env}^{\Gamma,\mu})\mathfrak F\Vert_{L^2}\leq k^2\Vert \mathfrak F\Vert_{L^2}+2k\Vert \nabla_\gamma^\Gamma\mathfrak F\Vert_{L^2(\R^d\times\Gamma\to\R^d;\hat\mu)}+k\big\Vert\vert H\vert_2\big\Vert_{L^\infty(\Gamma;\mu)}\Vert \mathfrak F\Vert_{L^2}
$$
holds for any $\mathfrak F\in D_k$. Since
\begin{multline*}
\Vert \nabla_\gamma^\Gamma \mathfrak F\Vert^2_{L^2(\R^d\times\Gamma\to\R^d;\hat\mu)}\leq \mathcal E_{env}^{\Gamma,\mu}(\mathfrak F,\mathfrak F)\\
=-(H_{env}^{\Gamma,\mu}\mathfrak F,\mathfrak F)_{L^2}\leq \Vert H_{env}^{\Gamma,\mu}\mathfrak F\Vert_{L^2}\Vert \mathfrak F\Vert_{L^2}\leq \frac{1}{2}\left(\frac{1}{C}\Vert H_{env}^{\Gamma,\mu}\mathfrak F\Vert_{L^2}+C \Vert \mathfrak F\Vert_{L^2}\right)^2
\end{multline*}
for any $0<C<\infty$, the claimed Kato-boundedness follows.
\end{proof}

\begin{theorem}\label{thm:uniqueness}
Let $\phi$ fulfill (SS), (LR), (I) and (D$\mbox{L}^p$) for some $p\in (d,\infty)\cap [2,\infty)$ and let $\mu\in\mathcal G_{ibp}^{gc}(\Phi_\phi,z\,\sigma_\phi)$. Then $H_{\phi}:=\langle\nabla\phi,\cdot\rangle\in L^p(\Gamma\to\R^d;\mu)$ and if $q\in [1,2]$ is such that $\frac{1}{q}=\frac{1}{2}+\frac{1}{p}$, the operator $(H_{coup}^{\R^d\times\Gamma,\hat\mu},D)=(L_{H_\phi},D)$ is essentially m-dissipative in $L^q(\R^d\times\Gamma;\hat\mu)$. 
\end{theorem}
\begin{proof}
For the first assertion see Lemma \ref{lem:integrability}. Let $\chi\in C_0^\infty(\R^d)$ be such that $1_{[-1,1]^d}\leq \chi\leq 1_{[-2,2]^d}$ and define $\chi_K:=\chi(K^{-1}\cdot)$ for $K\in\N$. Consider each $\chi_K$ as a function on $\R^d\times \Gamma$ by setting $\chi_K(\xi,\gamma):=\chi_K(\xi)$, $(\xi,\gamma)\in\R^d\times\Gamma$. For $n\in\N$, define $H_n:=H_{\phi}\cdot 1_{\{\vert H_\phi\vert_2\leq n\}}$. Then $H_n\in L^\infty(\Gamma\to\R^d;\mu)$. By dissipativity of $L_{H_\phi}$ it is sufficient to prove that $(1-L_{H_{\phi}})D$ is dense in $L^q$, and by denseness of $\hat D$ in $L^q$ this reduces to finding an approximate solution $\mathfrak F\in D$ of $(1-L_{H_\phi})\mathfrak F=\mathfrak G$ for any $\mathfrak G\in \hat D=\bigcup_{k\in\N}\hat D_k$. Thus, let $k\in\N$, $0\neq \mathfrak G\in \hat D_k$ and $\varepsilon>0$. For any $\mathfrak F\in \hat D_k$ and $K\in\N$ it holds
\begin{multline*}
\Vert (1-L_{H_\phi})(\chi_K\mathfrak F)-\mathfrak G\Vert_{L^q}\\
\leq \Vert \chi_K((1-L_{H_\phi})\mathfrak F-\mathfrak G)\Vert_{L^q}+\delta_K+2\Vert (\nabla_\xi\chi_K,(\nabla_\gamma^\Gamma-\nabla_\xi)\mathfrak F)_{\R^d}\Vert_{L^q}\\
+\Vert \mathfrak F (H_\phi,\nabla_\xi \chi_K)_{\R^d}\Vert_{L^q}+\Vert \mathfrak F\cdot\Delta_\xi \chi_K\Vert_{L^q},
\end{multline*}
where $\delta_K:=\Vert (1-\chi_K)\mathfrak G\Vert_{L^q}$. Setting $\theta_K:=\Vert \Delta_\xi \chi_K\Vert_{L^p(\R^d)}+\left\Vert \vert \nabla\chi_K\vert_2\right\Vert_{L^p(\R^d)}$ and using the H\"older inequality, we find that the right-hand side can be estimated by
\begin{multline*}
\Vert \chi_K ((1-L_{H_\phi})\mathfrak F-\mathfrak G)\Vert_{L^q}+\delta_K\\
+\theta_K\left(2\Vert (\nabla_\gamma^\Gamma-\nabla_\xi)\mathfrak F\Vert_{L^2(\R^d\times\Gamma\to\R^d;\hat\mu)}+\left\Vert \vert H_{\phi}\vert_2\right\Vert_{L^p(\Gamma;\mu)}\Vert \mathfrak F\Vert_{L^2}+\Vert \mathfrak F\Vert_{L^2}\right).
\end{multline*}
Let $n\in\N$. Note that by dissipativity of $L_{H_n}$, we have $\Vert \mathfrak F\Vert_{L^2}\leq \Vert (1-L_{H_n})\mathfrak F\Vert_{L^2}$ and by antisymmetry of $L_{H_n}-H_{coup}^{\R^d\times\Gamma,\hat\mu}$ it follows
\begin{multline*}
\Vert (\nabla_\gamma^\Gamma-\nabla_\xi)\mathfrak F\Vert^2_{L^2(\R^d\times\Gamma\to\R^d;\hat\mu)}\leq \mathcal E_{coup}^{\R^d\times\Gamma,\hat\mu}(\mathfrak F,\mathfrak F)\leq ((1-H_{coup}^{\R^d\times\Gamma,\hat\mu})\mathfrak F,\mathfrak F)_{L^2}\\
=((1-L_{H_n})\mathfrak F,\mathfrak F)_{L^2}\leq \Vert (1-L_{H_n})\mathfrak F\Vert_{L^2}\Vert \mathfrak F\Vert_{L^2}\leq \Vert (1-L_{H_n})\mathfrak F\Vert_{L^2}^2.
\end{multline*}
Setting $\tilde\theta_K:=\theta_K\,(3+\left\Vert \vert H_\phi\vert_2\right\Vert_{L^p(\Gamma;\mu)})$ we obtain
\begin{eqnarray*}
\lefteqn{\Vert (1-L_{H_\phi})(\chi_K\mathfrak F)-\mathfrak G\Vert_{L^q}}\\
& &\leq \Vert \chi_K((1-L_{H_\phi})\mathfrak F-\mathfrak G)\Vert_{L^q}+\delta_K+\tilde\theta_K\Vert (1-L_{H_{n}})\mathfrak F\Vert_{L^2}\\
& &\leq \Vert \chi_K((1-L_{H_{n}})\mathfrak F-\mathfrak G)\Vert_{L^q}+\Vert \chi_K\,\vert H_\phi-H_{n}\vert_2\Vert_{L^p}\Vert \nabla_\xi\mathfrak F\Vert_{L^2}+\delta_K+\tilde\theta_K\Vert (1-L_{H_n})\mathfrak F\Vert_{L^2}\\
& &\leq C_K\Vert (1-L_{H_n})\mathfrak F-\mathfrak G\Vert_{L^2}+\\
& &\quad\quad\quad\quad\quad k C_K\left\Vert \vert H_{\phi}-H_n\vert_2\right\Vert_{L^p(\Gamma;\mu)}\Vert (1-L_{H_n})\mathfrak F\Vert_{L^2}+\delta_K+\tilde\theta_K \Vert (1-L_{H_n})\mathfrak F\Vert_{L^2},
\end{eqnarray*}
where $C_K:=\Vert \chi_K\Vert_{L^p(\R^d)}$. For the last estimate above it is crucial to have $\mathfrak F\in \hat D_k$. Note that $\tilde\theta_K\to 0$ as $K\to\infty$. Hence, we may fix $K\in\N$ large enough such that $\delta_K\leq \varepsilon/4$ and $\tilde\theta_K\leq \frac{\varepsilon}{8\Vert \mathfrak G\Vert_{L^2}}$. Then we choose $n\in\N$ large enough such that $\left\Vert \vert H_\phi-H_n\vert_2\right\Vert_{L^p(\Gamma;\mu)}\leq \frac{\varepsilon}{8kC_K \Vert \mathfrak G\Vert_{L^2}}$. Finally, according to Lemma \ref{lem:uniquenessprep} we may choose $\mathfrak F\in \hat D_k$ such that $\Vert (1-L_{H_n})\mathfrak F-\mathfrak G\Vert_{L^2}\leq \frac{\varepsilon}{4 C_K}$ and $\Vert (1-L_{H_n})\mathfrak F\Vert_{L^2}\leq 2\Vert \mathfrak G\Vert_{L^2}$. It follows that
$$
\Vert (1-L_{H_\phi})(\chi_K\mathfrak F)-\mathfrak G\Vert_{L^q}\leq \varepsilon,
$$
and since $\chi_K \mathfrak F\in D$, the theorem is shown.
\end{proof}

\begin{remark}\label{rem:uniqueness}
Let $\mu\in\mathcal G_{ibp}^{gc}(\Phi_\phi,z\sigma_\phi)$ and consider the situation of Theorem \ref{thm:uniqueness} above. Let $\mathbf M^{env}$, $\mathbf M^{coup}$ be the processes from Section \ref{sec:existence}. Assume there is another conservative right process 
$$
{\mathbf M}=(\mathbf \Omega,\mathbf F,(\mathbf F_t)_{t\geq 0},(\mathbf X^{(1)}_t,\mathbf X^{(2)}_t)_{t\geq 0},(\mathbf P_{(\xi,\gamma)})_{(\xi,\gamma)\in \R^d\times\ddot\Gamma})
$$ 
(as e.g.~defined in \cite{MaRo92}) with state space $\R^d\times\ddot\Gamma$ and continuous paths, such that for Lebesgue a.e.~$\xi\in \R^d$ we have that $\mathbf P_{\delta_{\xi}\otimes \mu}\circ (\mathbf X^{(2)})^{-1}$, considered as law on $C([0,\infty),\ddot \Gamma)$, coincides with $\mathbf P_{\mu}^{env}$, and moreover $\mathbf P_{\delta_{\xi}\otimes\mu}$ solves the martingale problem for $L_{coup}^{\R^d\times\Gamma,\hat\mu}$ on $C_0^\infty(\R^d)\otimes \mathcal FC_b^\infty(\mathcal D,\Gamma)$. Then with some effort (applying also the proof of Lemma \ref{lem:identification2} to $\mathbf M$) one shows that $d\xi\otimes\mu=\hat \mu$ is an invariant measure for $\mathbf M$ and that e.g.~for some bounded probability density $h\in L^1(\R^d)$ the law $\mathbf P_{hd\xi\otimes \mu}$ solves the martingale problem for $C_0^\infty(\R^d)\otimes D(H_{env}^{\Gamma,\mu})_b$. By invariance the transition semigroup associated with $\mathbf M$ gives rise to a strongly continuous contraction semigroup in each $L^r(\R^d\times\Gamma;\hat\mu)$, $r\in [1,\infty)$, and by the arguments used in the proof of \cite[Theorem 3.5]{AR95} and by Theorem \ref{thm:uniqueness} this semigroup is equal to $(T_t^{coup})_{t\geq 0}$, hence $\mathbf P_{(\xi,\gamma)}=\mathbf P^{coup}_{(\xi,\gamma)}$ for $\hat\mu$-a.e.~$(\xi,\gamma)\in\R^d\times\Gamma$. Thus the above uniqueness result shows that, \emph{given the environment process is} $\mathbf M^{env}$, there is only one possible coupled process. Whether there is also only one possible environment process is another - much more difficult - question, which we do not attack here.
\end{remark}

\subsection{Proof of \eqref{eqn:decomposition}}\label{sub:decomposition}

Let us at first consider a modified setting in which \eqref{eqn:decomposition} is immediate: In this modified setting the tagged particle is only allowed to move in the cube $(-\kappa,\kappa]^d$, $\kappa\in\N$, with periodic boundary, i.e.~in the $d$-dimensional torus. (In fact, we do not consider the corresponding stochastic dynamics, but instead stay on the level of functional analytic objects.) We start from the nonnegative definite bilinear form $(\mathcal E_{coup,\kappa},C_{per}^{\infty}([-\kappa,\kappa]^d)\otimes \mathcal FC_b^\infty(\mathcal D,\Gamma))$ on $L^2([-\kappa,\kappa]^d\times\Gamma;d\xi\otimes \mu)=:L^{2,\kappa}$, given by
\begin{multline*}
\mathcal E_{coup,\kappa}(\mathfrak F,\mathfrak G):=\int_{\Gamma}\int_{[-\kappa,\kappa]^d} \left(\nabla^\Gamma \mathfrak F(\xi,\gamma),\nabla^\Gamma \mathfrak G(\xi,\gamma)\right)_{T_{\gamma}\Gamma} d\xi d\mu(\gamma)\\
+\int_{\Gamma}\int_{[-\kappa,\kappa]^d} \left((\nabla_\gamma^\Gamma-\nabla_\xi)\mathfrak F(\xi,\gamma),(\nabla_\gamma^\Gamma-\nabla_\xi)\mathfrak G(\xi,\gamma)\right)_{\R^d}\,d\xi d\mu(\gamma)
\end{multline*}
for $\mathfrak F, \mathfrak G\in C_{per}^\infty([-\kappa,\kappa]^d)\otimes \mathcal FC_b^\infty(\mathcal D,\Gamma)$. Here $C_{per}^{\infty}([-\kappa,\kappa]^d)$ denotes the restrictions of $2\kappa$-periodic (in all arguments), infinitely often differentiable  functions on $\R^d$ to the cube $[-\kappa,\kappa]^d$. As in the proof of \cite[Theorem 4.15]{FaGr08} one obtains for any $f\otimes F\in C_{per}^\infty([-\kappa,\kappa]^d)\otimes \mathcal FC_b^\infty(\mathcal D,\Gamma)$, $\kappa\in\N$, that
$$
\mathcal E_{coup,\kappa}(\mathfrak F,\mathfrak G)=-(L_{coup,\kappa}\mathfrak F,\mathfrak G)_{L^{2,\kappa}},
$$
where $L_{coup,\kappa}=L_{env}^{\Gamma,\mu}+\Delta_\xi-2(\nabla_\gamma^\Gamma,\nabla_\xi)_{\R^d}+(\langle \nabla\phi,\cdot\rangle,\nabla_\xi)_{\R^d}$. (For proving this, periodicity of the functions in the first argument is needed to ensure that the integrations by parts do not produce boundary terms.) It follows from \cite[Proposition I.3.3]{MaRo92} that the form $\mathcal E_{coup,\kappa}$ is closable and its generator $(H_{coup,\kappa},D(H_{coup,\kappa}))$ is the Friedrichs extension of the operator $(L_{coup,\kappa},C_{per}^{\infty}([-\kappa,\kappa]^d)\otimes \mathcal FC_b^\infty(\mathcal D,\Gamma))$ defined above. It is not difficult to verify that the closure of $\mathcal E_{coup,\kappa}$ is a Dirichlet form, which implies that the associated strongly continuous contraction semigroup $(T_{t,2}^{coup})_{t\geq 0}$, given by $T_{t,2}^{coup}:=\exp(t H_{coup,\kappa})$, $t\geq 0$, is sub-Markovian and gives rise to contraction semigroups $(T_{t,q}^{coup,\kappa})_{t\geq 0}$ on $L^{q,\kappa}:=L^q([-\kappa,\kappa]^d\times \Gamma;d\xi\otimes \mu)$ (with analogous properties as the semigroups $(T_{t,q}^{coup})_{t\geq 0}$), $q\in [1,\infty]$. Analogously to the proof of Lemma \ref{lem:preparationdomain} we have $C_{per}^\infty([-\kappa,\kappa]^d)\otimes D(H_{env}^{\Gamma,\mu})_b\subset D(H_{coup,\kappa})$ and $H_{coup,\kappa}$ has on this set the same form as $L_{coup,\kappa}$. For $F\in \mathcal FC_b^\infty(\mathcal D,\Gamma)$ it is easily seen that $t\mapsto 1\otimes T_{t,2}^{env}F$ is a solution of the abstract Cauchy problem for $H_{coup,\kappa}$ with initial value $1\otimes F$ as well as $t\mapsto T_{t,2}^{coup}(1\otimes F)$. Well-posedness of this abstract Cauchy problem implies the following lemma.

\begin{lemma}\label{lem:decompositionsimple}
For any $F\in L^2(\Gamma;\mu)$ and $t\geq 0$ it holds
$$
T_{t,2}^{coup,\kappa}(1\otimes F)=1\otimes T_{t,2}^{env}F.
$$
\end{lemma}

Note that for the above conclusion it is crucial that $1\otimes F$ is square integrable in the modified setting (in $L^{2,\kappa}$). In the original setting (in $L^2$) this is not the case, which is why \eqref{eqn:decomposition} is not immediate.

We now use the classical semigroup convergence theorem. For convenience we use the formulation from \cite{Tro58}, \cite{Ku69}: Let $q\in [1,2]$. For $\kappa\in\N$ define $P_\kappa: L^q\to L^{q,\kappa}$ by $P_\kappa \mathfrak F:=\mathfrak F|_{[-\kappa,\kappa]^d\times \Gamma}$, $\mathfrak F\in L^q$. Since $\Vert P_\kappa \mathfrak F\Vert_{L^{q,\kappa}}\to \Vert \mathfrak F\Vert_{L^q}$ as $\kappa\to\infty$, we have $L^{q,\kappa}\to L^q$ in the sense of the abovementioned papers. One says that a sequence $(\mathfrak F_\kappa)_{\kappa\in\N}$ with $\mathfrak F_\kappa\in L^{q,\kappa}$, $\kappa\in\N$, converges to some $\mathfrak F\in L^q$, iff $\lim_{\kappa\to\infty}\Vert P_\kappa \mathfrak F-\mathfrak F_\kappa\Vert_{L^{q,\kappa}}=0$. A uniformly bounded sequence $(A_\kappa)_{\kappa\in\N}$ of bounded linear operators $A_\kappa$ on $L^{q,\kappa}$, $\kappa\in\N$, is said to converge to a bounded linear operator $A$ on $L^q$, if $A_\kappa f_\kappa$ converges to $Af$ as $\kappa\to\infty$ for all sequences $(f_\kappa)_{\kappa\in\N}$ converging to $f$ as described above.

\begin{lemma}\label{lem:convergenceofsemigroups}
Assume that $(H_{coup}^{\R^d\times\Gamma,\hat\mu},C_0^\infty(\R^d)\otimes D(L_{env})_b)$ is essentially m-dissipative in $L^q$ for some $q\in [1,2]$. Then for $t\geq 0$ we have $T_{t,q}^{coup,\kappa}\to T_{t,q}^{coup}$ in the abovedescribed sense as $\kappa\to\infty$.
\end{lemma}
\begin{proof}
For any $\mathfrak F\in C_0^\infty(\R^d)\otimes D(H_{env}^{\Gamma,\mu})_b$ there exists $\kappa_0\in\N$ such that $P_\kappa\mathfrak F\in C_0^\infty((-\kappa,\kappa)^d))\otimes D(H_{env}^{\Gamma,\mu})\subset D(H_{coup,\kappa})_b$ and $H_{coup,\kappa}P_\kappa \mathfrak F=P_\kappa H_{coup}^{\R^d\times\Gamma,\hat\mu}\mathfrak F$ for all $\kappa\geq \kappa_0$. Thus $(H_{coup,\kappa}P_\kappa\mathfrak F)_{\kappa\geq \kappa_0}$ converges to $H_{coup}^{\R^d\times \Gamma,\hat\mu}\mathfrak F$ as $\kappa\to\infty$ along $L^{q,\kappa}\to L^q$. I.e., the extended limit of the sequence of the operators $(H_{coup,\kappa},D(H_{coup,\kappa}))$, $\kappa\in\N$, extends the essentially m-dissipative operator $(H_{coup}^{\R^d\times\Gamma,\hat\mu},C_0^\infty(\R^d)\otimes D(H_{env}^{\Gamma,\mu})_b)$ and the assertion follows from \cite[Theorem 2.1]{Ku69} (where the reader also finds the definition of the ``extended limit'').
\end{proof}

\begin{theorem}\label{thm:decomposition}
Let $\phi$ fulfill the assumptions (SS), (LR), (I) and (D$\mbox{L}^p$) for some $p\in (d,\infty)\cap [2,\infty)$ and let $\mu\in\mathcal G_{ibp}^{gc}(\Phi_\phi,z\,\sigma_\phi)$. Then for any $F\in L^\infty(\Gamma;\mu)$ and any $t\geq 0$ equation \eqref{eqn:decomposition} is valid.
\end{theorem}
\begin{proof}
We denote by $\langle \cdot,\cdot\rangle_\kappa$ the dualization between $L^{1,\kappa}$ and $L^{\infty,\kappa}$, $\kappa\in\N$, and by $\langle \cdot,\cdot\rangle_\infty$ the dualization between $L^1$ and $L^\infty$. Let w.l.o.g.~$0\leq F\leq 1$. Choose $(\chi_n)_{n\in\N}\subset C_0^\infty(\R^d)$ such that $0\leq \chi_n\uparrow 1$ as $n\to\infty$. Then by Lemma \ref{lem:convergenceofsemigroups} and Theorem \ref{thm:uniqueness} we have for some $q\in [1,2]$
$$
\Vert T_{t,q}^{coup,\kappa}P_\kappa (\chi_n\otimes F)-T_{t,q}^{coup}(\chi_n\otimes F)\Vert_{L^q}\to 0
$$
as $\kappa\to\infty$ for all $n\in\N$. Here we extend $T_{t,q}^{coup,\kappa}P_\kappa (\chi_n\otimes F)$ to $\R^d$ by setting it to $0$ on $(\R^d\setminus [-\kappa,\kappa])\times\Gamma$. It follows that for any nonnegative $\mathfrak G\in L^1\cap L^\infty$ and $n\in\N$ it holds
$$
\langle P_\kappa \mathfrak G,T_{t,q}^{coup,\kappa}P_\kappa(\chi_N\otimes F)\rangle_\kappa\to \langle \mathfrak G,T_{t,q}^{coup}(\chi_\kappa\otimes F)\rangle_\infty
$$
as $\kappa\to\infty$. Thus
$$
\liminf_{\kappa\to\infty} \langle P_\kappa\mathfrak G,T_t^{coup,\kappa} (1\otimes F)\rangle_\kappa\geq \langle \mathfrak G,T_t^{coup}(1\otimes F)\rangle_\infty.
$$
Since $(T_{t,\infty}^{coup})_{t\geq 0}$ is conservative as well as $(T_{t,2}^{coup,\kappa})_{t\geq 0}$, by replacing $F$ by $1-F$ we obtain
$$
\limsup_{\kappa\to\infty} \langle P_\kappa\mathfrak G,T_t^{coup,\kappa}(1\otimes F)\rangle_\kappa\leq \langle \mathfrak G,T_t^{coup}(1\otimes F)\rangle_\infty,
$$
hence the $\liminf$ and $\limsup$ coincide and we have convergence. Choosing a function $\mathfrak G=g\otimes G$, with nonnegative $g\in C_0(\R^d)$ and $G\in L^\infty(\Gamma;\mu)$, and applying Lemma \ref{lem:decompositionsimple}, we conclude that
$$
\langle \mathfrak G,T_t^{coup}(1\otimes F)\rangle_\infty=\lim_{\kappa\to\infty}\int_{[-\kappa,\kappa]^d} g\,d\xi\,(G,T_{t,2}^{env}F)_{L^2(\Gamma;\mu)}=\langle \mathfrak G,1\otimes T_{t,2}^{env}F\rangle_\infty.
$$
Since the linear span of functions $\mathfrak G$ as above is dense in $L^1$, the assertion follows.
\end{proof}

\section{Representation of the displacement of the tagged particle in terms of the environment process}\label{sec:displacement}

The aim of this section is to prove (under the assumptions stated below) that all information on the evolution of the tagged particle with the exception of its initial position, can be obtained from the environment process, at least if suitable initial distributions are chosen for both processes. We make this precise: Let $0\leq h\in L^1(\R^d)\cap L^\infty(\R^d)$ be a probability density w.r.t.~Lebesgue measure. For $\mu\in\mathcal G_{ibp}^{gc}(\Phi_\phi,z\,\sigma_\phi)$ we define the probability measure $h\hat\mu:=(h\,d\xi)\otimes \mu$ on $\R^d\times\Gamma$. Then $\mathbf P_{h\hat\mu}^{coup}(A):=\int_{\R^d\times\Gamma} \mathbf P_{(\xi,\gamma)}^{coup}(A) d(h\hat\mu)(\xi,\gamma)$, $A\in\mathbf \Omega^{coup}$, can be considered as a law on $C([0,\infty);\R^d\times\ddot\Gamma)$. The analogously defined law $\mathbf P_\mu^{env}$ can be considered as a law on $C([0,\infty);\ddot\Gamma)$. Paths from $C([0,\infty);\R^d\times\ddot\Gamma)$ and $C([0,\infty;\ddot\Gamma)$ are denoted below by $(\xi_t,\gamma_t)_{t\geq 0}$, $(\gamma_t)_{t\geq 0}$, respectively. Our aim is to prove the following result.

\begin{theorem}\label{thm:displacement}
Let $\phi$ fulfill (SS), (LR), (I) and (D$\mbox{L}^p$) for some $p\in (d,\infty)\cap [2,\infty)$, let $\mu\in \mathcal G_{ibp}^{gc}(\Phi_\phi,z\,\sigma_\phi)$, and let $\mathbf P_{h\hat\mu}^{coup}$ and $\mathbf P_{\mu}^{env}$ be laws on $C([0,\infty);\R^d\times\ddot\Gamma)$, $C([0,\infty);\ddot\Gamma)$, respectively, as described above. Then there exists a $C([0,\infty);\R^d)$-valued random variable $(\eta_t)_{t\geq 0}$ on $C([0,\infty);\ddot\Gamma)$ with the following properties:
\begin{enumerate}
\item For all $0\leq a<b$ the random variable $\eta_b-\eta_a$ coincides $\mathbf P^{env}_\mu$-a.s.~with a random variable that is measurable w.r.t.~the $\sigma$-field generated by $\{\gamma_s\,|\,a\leq s\leq b\}$. In particular, $(\eta_t)_{t\geq 0}$ is adapted to the $\mathbf P^{env}_\mu$-completion of the filtration $(\sigma(\gamma_s\,|\,0\leq s\leq t))_{t\geq 0}$ relative to the Borel $\sigma$-field on $C([0,\infty);\ddot\Gamma)$.
\item The distribution of the paths $(\eta_t,\gamma_t)_{t\geq 0}$ under $\mathbf P_{\mu}^{env}$ coincides with the distribution of $(\xi_t-\xi_0,\gamma_t)_{t\geq 0}$ under $\mathbf P_{h\hat\mu}^{coup}$.
\end{enumerate}
\end{theorem}
\begin{proof}
This theorem is a consequence of Lemmas \ref{lem:identification1} and \ref{lem:identification2} below. More precisely, Lemma \ref{lem:identification2} gives for each $0\leq a<b$ an $\R^d$-valued random-variable $X_{[a,b]}$ on $C([0,\infty);\R^d\times\ddot\Gamma)$ which is a function of $(\gamma_s)_{a\leq s\leq b}$ and coincides $\mathbf P_{h\hat\mu}^{coup}$-a.s.~with $\xi_b-\xi_a$. This random variable $X_{[a,b]}$ can therefore be defined as a measurable function $X_{[a,b]}=X_{[a,b]}((\gamma_s)_{a\leq s\leq b})$ on $C([0,\infty);\ddot\Gamma)$, and setting $\tilde \eta_t:=X_{[0,t]}$, we find by Lemma \ref{lem:identification1} that for each $t$ the distribution of $(\eta_t,\gamma_t)$ under $\mathbf P_{\mu}^{env}$ coincides with the distribution of $(\xi_t-\xi_0,\gamma_t)$ under $\mathbf P_{h\hat\mu}^{coup}$. In particular $(\tilde \eta_t)_{t\geq 0}$ is $\mathbf P_\mu^{env}$-a.s.~continuous, so the desired family $(\eta_t)_{t\geq 0}$ can be chosen as a continuous version of $(\tilde \eta_t)_{t\geq 0}$.
\end{proof}

\begin{lemma}\label{lem:identification1}
Under the assumptions of Theorem \ref{thm:displacement}, the image law $\mathbf P_{h\hat\mu}^{coup,2}$ of $\mathbf P_{h\hat\mu}^{coup}$ under the projection $C([0,\infty);\R^d\times\ddot\Gamma)\ni (\xi_t,\gamma_t)_{t\geq 0}\mapsto (\gamma_t)_{t\geq 0}\in C([0,\infty);\ddot\Gamma)$ coincides with $\mathbf P_{\mu}^{env}$.
\end{lemma}
\begin{proof}
It suffices to compare the finite dimensional distributions of both laws. Let $F_1,\cdots,F_k: \ddot\Gamma\to \R$ be bounded $\mathcal B(\ddot\Gamma)$-measurable functions and let $0\leq t_1<\cdots t_k<\infty$, $k\in\N$. Then by Theorem \ref{thm:decomposition}
\begin{eqnarray*}
\lefteqn{\mathbf E_{h\hat\mu}^{coup,2}[F_1(\gamma_{t_1})\cdots F_k(\gamma_{t_k})]}\\
& &=\int_{\R^d\times\Gamma} T_{t_1}^{coup}\left((1\otimes F_1)\cdot T_{t_2-t_1}^{coup}\left(\cdots \left((1\otimes F_{k-1})\cdot T_{t_k-t_{k-1}}^{coup} (1\otimes F_k)\right)\cdots\right)\right)\,d(h\hat\mu)\\
& &=\int_{\Gamma} T_{t_1}^{env} \left(F_1\cdot T_{t_2-t_1}^{env}\left(\cdots \left(F_{k-1}\cdot T^{env}_{t_k-t_{k-1}}F_k\right)\cdots\right)\right)\,d\mu=\mathbf E_{\mu}^{env}[F_1(\gamma_{t_1})\cdots F_k(\gamma_{t_k})],
\end{eqnarray*}
which shows the assertion.
\end{proof}
The above lemma allows us to work only with the coupled process for constructing the displacement of the tagged particle from the environment process. For doing so we need some technical preparations: At first we note that the (D$\mbox{L}^p$)-assumption implies that
$$
q_n:=n^{-1}\left(1+\int_0^n \int_{\R^d\setminus [-r,r]^d} \vert \nabla\phi(x)\vert_2\,e^{-\phi(x)}\,dx\,dr\right)\to 0
$$
as $n\to\infty$. Thus, setting $r_n:=\sqrt{q_n}$, $n\in\N$, it follows that $r_n=q_n/r_n\to 0$ as $n\to\infty$. For $n\in\N$ set $c_n^0:=1_{[-n,n]^d}: \R^d\to\R$. Then $\langle c_n^0,\cdot\rangle: \Gamma\to\R$ counts the points in $[-n,n]^d$ of a given configuration.

\begin{lemma}\label{lem:aslemma}
Let $\mu\in \mathcal G_{ibp}^{gc}(\Phi_\phi,z\,\sigma_\phi)$. Then
\begin{equation}\label{eqn:wachstum}
\lim_{n\to\infty} \frac{\langle c_n^0,\gamma\rangle}{r_n n^d}=\infty\quad \mbox{for $\mu$-a.e.~$\gamma\in\Gamma$.}
\end{equation}
\end{lemma}
\begin{proof}
By Remark \ref{rem:measuresmod} and Lemma \ref{lem:measures} we only need to show the assertion for measures $\mu\in \textnormal{ex}\,\mathcal G_{\theta}^{gc}(\Phi_\phi,z\,dx)$. For any such measure we have $\int_{\Gamma} \langle 1_{[0,1]^d},\cdot\rangle\,d\mu=:\rho>0$. Using the multidimensional ergodic theorem (one can e.g.~directly apply \cite[Theorem 4]{Pit42}), one finds that
$$
f_n(\gamma):= \frac{1}{(2n)^d}\langle 1_{[-n,n]^d},\gamma\rangle\to \hat f(\gamma)
$$
as $n\to\infty$ where $\hat f$ is the conditional expectation of $f=\langle 1_{[0,1]^d},\cdot\rangle$ w.r.t.~the sub-$\sigma$-field of $\mathcal B(\Gamma)$ of the sets which are $\mu$-a.e.~invariant w.r.t.~translations in $\R^d$. By triviality of $\mu$ on this $\sigma$-field we obtain $\hat f(\gamma)=\rho$ $\mu$-a.s., and \eqref{eqn:wachstum} follows.
\end{proof}

\begin{lemma}\label{lem:Ylemma}
For $1\leq i\leq d$ and $n\in\N$ we define a function $\widetilde Y^n_i: \Gamma\to\R$ by
\begin{multline*}
\widetilde Y^n_i(\gamma):=-\frac{1}{r_n n^d +\langle c_n^0,\gamma\rangle}\Bigg(\langle \partial_i\phi\, c_n^0,\gamma\rangle+\langle \partial_i\phi,\gamma\rangle\langle c_n^0,\gamma\rangle\\
+\sum_{\{x,y\}\subset\gamma}\partial_i \phi(x-y) (c_n^0(x)-c_n^0(y))\Bigg)+\langle \partial_i\phi,\gamma\rangle,
\end{multline*}
for $\gamma\in\Gamma$ such that all sums converge absolutely (which is $\mu$-a.e.~true). Then $\widetilde Y_i^n\to 0$ in $L^1(\Gamma;\mu)$ as $n\to\infty$.
\end{lemma}
\begin{proof}
We first note that for $n\in\N$ it holds $\mu$-a.e.
\begin{equation}\label{eqn:urgh}
\left\vert \frac{1}{r_n n^d+\langle c_n^0,\cdot\rangle}\langle \partial_i\phi c_n^0,\cdot\rangle\right\vert\leq \frac{\langle \vert\partial_i\phi\vert,\cdot\rangle}{1+\langle c_n^0,\cdot\rangle}.
\end{equation}
Since the numerator on the right-hand side is integrable and by Lemma \ref{lem:aslemma} the denominator converges $\mu$-a.s.~to $\infty$ as $n\to\infty$, we obtain convergence of \eqref{eqn:urgh} to $0$ in $L^1(\Gamma;\mu)$. Moreover, Lemma \ref{lem:aslemma} and $\mu$-integrability of $\langle \partial_i\phi,\cdot\rangle$ imply
$$
\left(1-\frac{\langle c_n^0,\cdot\rangle}{r_n n^d+\langle c_n^0,\cdot\rangle}\right)\langle \partial_i \phi,\cdot\rangle\to 0
$$
in $L^1(\Gamma;\mu)$ as $n\to\infty$, so we are left to prove convergence to $0$ of
\begin{multline}\label{eqn:lastYestimate}
\frac{1}{r_n n^d+\langle c_n^0,\cdot\rangle} \left\vert \sum_{\{x,y\}\subset\cdot} \partial_i \phi(x-y) (1_{[-n,n]^d}(x)-1_{[-n,n]^d}(y))\right\vert\\
\leq \frac{1}{r_n n^d} \sum_{\{x,y\}\subset \cdot} \left(1_{[-n,n]^d}(x)1_{\R^d\setminus [-n,n]^d}(y)+1_{[-n,n]^d}(y)1_{\R^d\setminus [-n,n]^d}(x)\right)\vert \partial_i \phi(x-y)\vert.
\end{multline}
Since $\mu$ fulfills an improved Ruelle bound (see Section \ref{sec:conditions}), there exists a constant $C<\infty$ such that the $L^1(\Gamma;\mu)$-norm of the right-hand side can be estimated by
\begin{eqnarray*}
\lefteqn{C\int_{\R^d\times\R^d} 1_{[-n,n]^d}(x)1_{\R^d\setminus [-n,n]^d}(y) \,\vert \partial_i \phi(x-y)\vert e^{-\phi(x-y)}\,dx\,dy}\\
& &\leq C\int_{[-n,n]^d}\int_{\R^d\setminus [-n,n]^d} \vert \partial_i\phi(x-y)\vert e^{-\phi(x-y)}\,dx\,dy\\
& &\leq 2dC (2n)^{d-1}\int_0^n \int_{\R^d\setminus [-r,r]^d}\vert \partial_i\phi(y)\vert e^{-\phi(y)}\,dy\,dr.
\end{eqnarray*}
(Here we again used \cite[Theorem 4.1]{KK02}.) It follows that the right-hand side of \eqref{eqn:lastYestimate} converges to $0$ in $L^1(\Gamma;\mu)$. This completes the proof.
\end{proof}

\begin{lemma}\label{lem:identification2}
Under the assumptions of Theorem \ref{thm:displacement}, for $0\leq a<b$, $\xi_b-\xi_a$ is $\mathbf P^{coup}_{h\hat\mu}$-a.s.~$\sigma(\{\gamma_s,|,a\leq s\leq b\})$-measurable (i.e.~coincides except on a null set with a $\sigma(\{\gamma_s\,|\,a\leq s\leq b\})$-measurable random variable).
\end{lemma}
\begin{proof}
In \cite[p.~849]{DeM89} it is proposed to obtain the stochastic part of the uniform motion of the environment from a solution of \eqref{eqn:env} by taking the limit $\lim_{N\to\infty}\frac{1}{N}\sum_{i=1}^N (B_t^i-B_t^0)$ (which together with the fact that the drift part is clearly measurable w.r.t.~the environment shows the desired adaptedness/measurability property). We emphasize that this idea cannot be applied directly in the present setting, since the dynamics is not obtained as $(\xi_t)_{t\geq 0}$, $(y^i_t)_{t\geq 0}$ fulfilling \eqref{eqn:env}, but as a solution of the martingale problem for the generator corresponding formally to the equations \eqref{eqn:tag}, \eqref{eqn:env}. However, the general idea, namely obtaining the uniform motion of the environment by averaging (namely over the particles contained in subsets of $\R^d$) is applicable in this setting, but leads to the following proof, which is a bit more involved than one might expect. 

Let $1\leq i\leq d$. For $n\in\N$ and $\delta>0$ choose functions $f_n^\delta, c_n^\delta\in C_0^\infty(\R^d)$ such that $f_n^\delta(x)=x_i$ and $c_n^\delta(x)=1$ for $x=(x_1,\cdots,x_d)\in [-n,n]^d$, $f_n^\delta(x)=c_n^\delta(x)=0$ for $x\in\R^d\setminus [-n-\delta,n+\delta]^d$ and set
$$
F_n^\delta:=\frac{\langle f_n^\delta,\cdot\rangle}{r_n n^d+\langle c_n^\delta,\cdot\rangle}.
$$
Then $F_n^\delta\in\mathcal FC_b^\infty(\mathcal D,\Gamma)$. It follows from Theorem \ref{thm:martingallemma}(i) that
$$
M_t^{n,\delta}:=F_n^\delta(\gamma_t)-F_n^\delta(\gamma_0)+\xi_{i,t}-\xi_{i,0}-\int_0^t L_{env}^{\Gamma,\mu}F_n^\delta(\gamma_s)+\langle \partial_i\phi,\gamma_s\rangle\,ds,
$$
$t\geq 0$, defines a continuous local martingale under $\mathbf P_{h\hat\mu}^{coup}$ with quadratic variation process given by
$$
\langle M^{n,\delta}\rangle_t=2\int_0^t (\nabla^\Gamma F_n^\delta(\gamma_s),\nabla^\Gamma F_n^\delta(\gamma_s))_{T_{\gamma_s}\Gamma}+\left(\nabla_\gamma^\Gamma F_n^\delta(\gamma_s)-e_i,\nabla_\gamma^\Gamma F_n^\delta(\gamma_s)-e_i\right)_{\R^d}\,ds,
$$
$t\geq 0$, where $e_i$ denotes the $i$-th unit vector in $\R^d$. Since by Lemma \ref{lem:identification1} and invariance of $\mu$ w.r.t.~$\mathbf M^{env}$ it holds
$$
\mathbf E_{h\hat\mu}^{coup}\left[\int_0^\infty e^{-t}Z(\gamma_s)\,ds\right]=\mathbf E_{\mu}^{env}\left[\int_0^\infty e^{-t}Z(\gamma_s)\,ds\right]\leq \int_0^\infty e^{-t}\Vert T_{t,1}^{env}Z\Vert_{L^1(\Gamma;\mu)}\,ds<\infty
$$ 
for any $Z\in L^1(\Gamma;\mu)$, we may replace the ``local'' in the previous statement by square integrable, see also Remark \ref{rem:martingale}. We now ignore the increments of $(F_n^\delta(\gamma_t))_{t\geq 0}$ at times $t\geq 0$ where $\gamma_t\cap A_n^\delta\neq \emptyset$ with $A_n^\delta:=[-n-\delta,n+\delta]^d\setminus (-n,n)^d$: Let $H_n^\delta: \Gamma\to \{0,1\}$ be the indicator function of the set $\{\gamma\in\Gamma\,|\,\gamma\cap A_n^\delta=\emptyset\}$ and set
$$
\widetilde F_t^{n,\delta}:=\int_0^t H_n^\delta(\gamma_s)dF_n^\delta(\gamma_s).
$$
Then
\begin{align}\label{eqn:themartingale}
\widetilde M_t^{n,\delta}&:=\int_0^t H_n^\delta(\gamma_s)\,dM^{n,\delta}_s\nonumber\\
&=\widetilde F_t^{n,\delta}+\int_0^t H_n^\delta(\gamma_s)\,d\xi_{i,s}-\int_0^t H_n^\delta(\gamma_s)\left(L_{env}^{\Gamma,\mu}F_n^\delta(\gamma_s)+\langle \partial_i\phi,\gamma_s\rangle\right)ds\nonumber\\
&=\widetilde F_t^{n,\delta}+\int_0^t H_n^\delta(\gamma_s)\,d\xi_{i,s}-\int_0^t H_n^\delta(\gamma_s)\widetilde Y_i^n(\gamma_s)\,ds,\quad t\geq 0,
\end{align}
defines a martingale under $\mathbf P_{h\hat\mu}^{coup}$ with quadratic variation process given by
\begin{eqnarray*}
\lefteqn{\langle \widetilde M^{n,\delta}\rangle_t=\int_0^t (H_n^\delta(\gamma_s))^2 d\langle M^{n,\delta}\rangle_s}\\
& &=2\int_0^t H_n^\delta (\gamma_s)\left(\left(\nabla^\Gamma F_n^\delta(\gamma_s),\nabla^\Gamma F_n^\delta(\gamma_s)\right)_{T_{\gamma_s}\Gamma}+\left(\nabla_\gamma^\Gamma F_n^\delta(\gamma_s)-e_i,\nabla_\gamma^\Gamma F_n^\delta(\gamma_s)-e_i\right)_{\R^d}\right)\,ds\\
& &=2\int_0^t H_n^\delta(\gamma_s)\left(\frac{\langle c_n^0,\gamma_s\rangle}{(r_n n^d+\langle c_n^0,\gamma_s\rangle)^2} +\left(\frac{\langle c_n^0,\gamma_s\rangle}{r_n n^d+\langle c_n^0,\gamma_s\rangle}-1\right)^2\right)\,ds,\quad\quad t \geq 0.
\end{eqnarray*}

Below we use the following (standard) notion of convergence for continuous processes: If $\mathbf P$ is a probability measure on a measurable space $(\Omega,\mathcal B)$ and $(Z_t^n)_{t\geq 0}$, $n\in\N$, and $(Z_t)_{t\geq 0}$ are $C([0,\infty);\R)$-valued random variables on $\Omega$, then we write $(Z_t^n)_{t\geq 0}\to (Z_t)_{t\geq 0}$ in $L^r(\mathbf P)$, $r\in [1,\infty]$, iff for any $T>0$ it holds $\int_{\Omega} \sup_{0\leq t\leq T} \vert Z_t^n-Z_t\vert^rd\mathbf P\to 0$ as $n\to\infty$. We consider the limits in this sense of the processes occurring in \eqref{eqn:themartingale}. First note that for any $T>0$ and any $Y\in L^1(\Gamma;\mu)$ it holds by Lemma \ref{lem:identification1} and invariance of $\mu$ w.r.t.~$\mathbf M^{env}$
\begin{align}\label{eqn:pathconvergence}
\mathbf E_{h\hat\mu}^{coup}\left[\int_0^T \vert (H_n^{\delta}(\gamma_s)-1)Y(\gamma_s)\vert \,ds\right]&=\int_0^T \mathbf E_{\mu}^{env}[\vert (H_n^\delta(\gamma_s)-1)Y(\gamma_s)\vert]\nonumber\\
&=T \Vert (H_n^\delta-1)Y\Vert_{L^1(\Gamma;\mu)}\to 0
\end{align}
as $\delta\to 0$ by Lebesgue's dominated convergence theorem, since $1-H_n^\delta$ decreases to the indicator function of the set $\{\gamma\in\Gamma\,|\,\gamma\cap \partial [-n,n]^d\neq \emptyset\}$ which has $\mu$-measure $0$. This already implies that $\left(\int_0^t H_n^\delta(\gamma_s) \widetilde Y_i^n(\gamma_s)\,ds\right)_{t\geq 0}\to \left(\int_0^t \widetilde Y_i^n(\gamma_s)\,ds\right)_{t\geq 0}$ in $L^1(\mathbf P_{h\hat\mu}^{coup})$. Since for $\delta>\delta'$ the martingale $\left(\int_0^t H_n^\delta(\gamma_s)\,d( M^{n,\delta}- M^{n,\delta'})_s\right)_{t\geq 0}$ is easily verified to be equal to $0$ $\mathbf P^{coup}_{h\hat\mu}$-a.s.~(by considering its quadratic variation process, which can be computed with the help of Theorem \ref{thm:martingallemma}(i)), by the Burkholder-Davis-Gundy inequality there is some $C<\infty$ such that for any $0<T<\infty$ and $\delta>\delta'$
\begin{eqnarray*}
\lefteqn{\mathbf E_{h\hat\mu}^{coup}\left[ \sup_{t\in [0,T]} \left\vert\int_0^t H_n^\delta(\gamma_s)d M^{n,\delta}_s-\int_0^t H_n^{\delta'}(\gamma_s)d M^{n,\delta'}_s\right\vert^2\right]}\\
& &=\mathbf E_{h\hat\mu}^{coup}\left[\sup_{t\in [0,T]} \left\vert\int_0^t (H_n^\delta(\gamma_s)- H_n^{\delta'}(\gamma_s))d M^{n,\delta'}_s\right\vert^2\right]\\
& &\leq C\, \mathbf E_{h\hat\mu}^{coup}\left[\int_0^T \vert H_n^\delta(\gamma_s)- H_n^{\delta'}(\gamma_s)\vert^2 \left(\frac{\langle c_n^0,\gamma_s\rangle}{(r_n n^d+\langle c_n^0,\gamma_s\rangle)^2} +\left(\frac{\langle c_n^0,\gamma_s\rangle}{r_n n^d+\langle c_n^0,\gamma_s\rangle}-1\right)^2\right)\,ds\right]\\
& &\leq 2C\,T\Vert H_n^\delta-H_n^{\delta'}\Vert_{L^1(\Gamma;\mu)}.
\end{eqnarray*}
This shows that $(\widetilde M^{n,\delta}_t)_{t\geq 0}$ converges as $\delta\to 0$ in $L^2(\mathbf P_{h\hat\mu}^{coup})$-sense (hence in $L^1(\mathbf P_{h\hat\mu}^{coup})$-sense) to some continuous process $(\widetilde M^n_t)_{t\geq 0}$, which one verifies to be a square integrable continuous martingale with quadratic variation process
$$
\langle \widetilde M^n\rangle_t=2\int_0^t \frac{\langle c_n^0,\gamma_s\rangle}{(r_n n^d+\langle c_n^0,\gamma_s\rangle)^2} +\left(\frac{\langle c_n^0,\gamma_s\rangle}{r_n n^d+\langle c_n^0,\gamma_s\rangle}-1\right)^2\,ds,\quad\quad t\geq 0.
$$
Similar arguments based on the semimartingale decomposition of $(\xi_{i,t})_{t\geq 0}$ (see Theorem \ref{thm:martingallemma}(ii)) show that $\left(\int_0^t H_n^\delta(\gamma_s)\,d\xi_{i,s}\right)_{t\geq 0}\to (\xi_{i,t}-\xi_{i,0})_{t\geq 0}$ in $L^1(\mathbf P_{h\hat\mu}^{coup})$. Hence $(\widetilde F^{n,\delta}_t)_{t\geq 0}$ converges in $L^1(\mathbf P_{h\hat\mu}^{coup})$ to some continuous process $(\widetilde F^n_t)_{t\geq 0}$ fulfilling
$$
\widetilde M^n_t=\widetilde F^{n}_t+\xi_{i,t}-\xi_{i,0}-\int_0^t \widetilde Y_n^i(\gamma_s)\,ds,\quad t\geq 0.
$$
Letting $n\to\infty$, it is shown analogously as \eqref{eqn:pathconvergence} using Lemma \ref{lem:Ylemma} that $\left(\int_0^t \widetilde Y_n^i(\gamma_s)\,ds\right)_{t\geq 0}$ converges to $0$ in $L^1(\mathbf P_{h\hat\mu}^{coup})$ and using Lemma \ref{lem:aslemma} we find that the same holds for $(\langle \widetilde M^n\rangle_t)_{t\geq 0}$. By the Burkholder-Davis-Gundy inequality the latter implies that $(\widetilde M^n_t)_{t\geq 0}\to 0$ in $L^2(\mathbf P_{h\hat\mu}^{coup})_{t\geq 0}$, and it follows that in $L^1$-sense $(-\widetilde F^n_t)_{t\geq 0}$ converges as $n\to\infty$ to $(\xi_{i,t}-\xi_{i,0})_{t\geq 0}$. A subsequence converges a.s.~and we denote the limit of such a subsequence by $(\eta_t)_{t\geq 0}$. In order to prove that $\eta_b-\eta_a$ is $\mathbf P_{h\hat\mu}^{coup}$-a.s.~$\sigma(\gamma_s\,|\,a\leq s\leq b)$-measurable, we only need to verify this measurability property for $\tilde F_b^{n,\delta}-\tilde F_a^{n,\delta}$. To verify this note first that $H_n^\delta$ is the indicator function of a closed set and can be represented as limit of continuous functions dominated by $1$, implying that the stochastic integral defining $\tilde F_b^{n,\delta}-\tilde F_a^{n,\delta}$ can be approximated in $\mathbf P^{coup}_{h\hat\mu}$-measure (and hence for a subsequence also a.s.) by stochastic integrals w.r.t.~continuous integrands (see e.g.~\cite[Theorem IV.2.12]{RY01}), which in turn can be approximated by Riemann-Stieltjes sums depending only on $\gamma_s$ with $s\in [a,b]$ (see e.g.~\cite[Proposition IV.2.13]{RY01}).

\end{proof}

\section{Ergodicity of the environment process}\label{sec:ergodicity}

As is well-known, in order to prove time-ergodicity (under suitable assumptions on $\mu$) of the law $\mathbf P_{\mu}^{env}$ defined in the previous section, we have to verify that the form $(\mathcal E_{env}^{\Gamma,\mu},D(\mathcal E_{env}^{\Gamma,\mu}))$ is irreducible, i.e.~that if $F\in D(\mathcal E_{env}^{\Gamma,\mu})$ fulfills $\mathcal E_{env}^{\Gamma,\mu}(F,F)=0$, then $F$ is constant. This property can be shown at least in two ways: Firstly, one could prove it for the closure of the bilinear form $(\mathcal FC_b^\infty(\mathcal D,\Gamma))^2\ni(F,G)\mapsto \mathcal \int_{\Gamma} (\nabla^\Gamma F,\nabla^\Gamma G)_{T_\cdot\Gamma}\,d\mu$. This can presumably be done along the lines of \cite[Theorem 6.6]{AKR98b} by proving at first an integration by parts characterization of the set of canonical Gibbs measures for $\phi$ with intensity measure $z\,e^{-\phi}dx$. The disadvantage of this method in the present setting is that one would have to assume that there exists $\mu\in\mathcal G_{ibp}^{gc}(\Phi_\phi,z\,\sigma_\phi)$ such that $\mu$ is extremal in the set of canonical Gibbs measures for $\phi$ with intensity measure $\sigma_\phi$. Even if all elements of $\textnormal{ex}\,\mathcal G_{t}^{gc}(\Phi_\phi,z\,\sigma_\phi)$ have this property (which is a natural assumption, see e.g.~\cite{AKR98b} and \cite[Theroem 6.14]{Ge79}), we would still need $(\textnormal{ex}\mathcal G_t^{gc}(\Phi_\phi,z\sigma_\phi))\cap \mathcal G_{ibp}^{gc}(\Phi_\phi,z\sigma_\phi)\neq \emptyset$, which seems not to be clear in general. The second way to prove irreducibility of $\mathcal E_{env}^{\Gamma,\mu}$, which is the way we use below, is to show this property for the closure of the form $(\mathcal E_{fro}^{\Gamma,\mu},\mathcal FC_b^\infty(\mathcal D,\Gamma))$, defined by
$$
\mathcal E_{fro}^{\Gamma,\mu}(F,G):=\int_{\Gamma}(\nabla_\gamma^\Gamma F,\nabla_\gamma^\Gamma G)_{\R^d}\,d\mu,
$$
$F,G\in \mathcal FC_b^\infty(\mathcal D,\Gamma)$. It is not difficult to see that for $\mu\in\mathcal G_{ibp}(\Phi_\phi,z\,\sigma_\phi)$ this form is closable and its closure $(\mathcal E_{fro}^{\Gamma,\mu},D(\mathcal E_{fro}^{\Gamma,\mu}))$ is a Dirichlet form. It corresponds to the dynamics of a frozen environment seen from a moving particle interacting with the environment via $\phi$. In particular, we cannot expect to obtain irreducibility of $\mathcal E_{fro}^{\Gamma,\mu}$ for the case when $d=1$ and $\phi$ is singular. (However, one should note that this case is not interesting concerning the derivation of an invariance principle for the tagged particle process using \cite{DeM89}, since diffusive scaling seems not appropriate in one-dimensional systems with collisions, see \cite{Ha65}.) On the other hand, in the following theorem we will see that in all other cases we obtain a quite natural - and nonvoid - irreducibility statement.

\begin{theorem}\label{thm:ergodicity}
Assume that $\phi$ fulfills (SS), (LR), (I) and (D$\mbox{L}^2$). Let $d\geq 2$ or $\phi$ be bounded. Assume that $\mu=\frac{1}{Z_{\mu_0}} e^{-\langle \phi,\cdot\rangle}\,d\mu_0$ for some $\mu_0\in \textnormal{ex}\,\mathcal G_{\theta}^{gc}(\Phi_\phi,z\,dx)$. Then $(\mathcal E_{fro}^{\Gamma,\mu},D(\mathcal E_{fro}^{\Gamma,\mu}))$ and hence $(\mathcal E_{env}^{\Gamma,\mu},D(\mathcal E_{env}^{\Gamma,\mu}))$ is irreducible. Thus $\mathbf P_{\mu}^{env}$ (as defined at the beginning of Section \ref{sec:displacement}) is time-ergodic.
\end{theorem}

\begin{remark}
If in addition to the assumptions in the above theorem we have $\nabla\phi\in L^1(\R^d\setminus B_r(0);dx)$ for some (hence all) $r>0$ (or $d=1$ and $\phi$ bounded), the assumption on $\mu$ can be replaced by the (in this case) equivalent and more natural assumption $\mu\in \textnormal{ex}\mathcal G_{ibp}^{gc}(\Phi_\phi,z\,\sigma_\phi)$, see the discussion at the end of Section \ref{sec:conditions}.
\end{remark}

\begin{proof}
Let $F\in D(\mathcal E_{fro}^{\Gamma,\mu})$ be such that $\mathcal E_{fro}^{\Gamma,\mu}(F,F)=0$. We have to prove that $F$ is constant. By the proof of \cite[Proposition I.4.17]{MaRo92} we may assume that $F$ is bounded. By Remark \ref{rem:measuresmod} (set $A=\{\gamma\in\Gamma\,|\,F(\gamma)>c\}$ for $c\in\R$) we only need to prove that for any $v\in\R^d$ it holds $F(\gamma+v)=F(\gamma)$ for $\mu$-a.e.~$\gamma\in\Gamma$.\\
We only prove this for $v=e_1$, the first standard unit vector of $\R^d$ and consider only the case $d\geq 2$, since the case $d=1$ and $\phi$ bounded is easier to treat. Note that by (D$\mbox{L}^1$) it is possible to fix a Lebesgue-version $\tilde\phi^+$ of $\phi^+$ and a Lebesgue nullset $N\subset\R^{d-1}$ such that for all $(x_2,\cdots,x_d)\in\R^{d-1}\setminus N$ the function $\tilde\phi^+(\cdot,x_2,\cdots,x_d)$ is weakly differentiable and continuous. Setting $\tilde\phi^+_{\sup}(x_1,\cdots,x_d):=\sup_{s\in [0,1]}\tilde\phi^+(x_1+s,x_2,\cdots,x_d)$, $(x_1,\cdots,x_d)\in\R^d$, we therefore obtain a Lebesgue-a.e.~finite function $\tilde\phi^+_{\sup}$. Moreover, since for any element $g$ of the Sobolev space $W^{1,1}(0,1)$ it holds $\sup_{s\in [0,1]}\vert g(s)\vert\leq \Vert g\Vert_{W^{1,1}(0,1)}$, we obtain for any $K\in\N$
\begin{eqnarray*}
\lefteqn{\int_{\R^d} \tilde\phi_{\sup}^+\wedge K\,dx=\int_{\R^{d}} \sup_{s\in [0,1]} (\tilde\phi^+\wedge K)(x_1+s,x_2,\cdots,x_d)\,dx_1dx_2\cdots dx_d}\\
& &\hspace{-0.25cm}\leq \int_{\R^{d}} \int_0^1 \vert (\phi^+\wedge K)(x_1+s,x_2,\cdots,x_d)\vert+\vert\partial_{x_1} (\phi^+\wedge K)(x_1+s,x_2,\cdots,x_d)\vert\,ds\,dx_1\cdots dx_d\\
& &\hspace{-0.25cm}=\Vert \phi^+\wedge K\Vert_{L^1(\R^d)}+\Vert \partial_{x_1}(\phi^+\wedge K)\Vert_{L^1(\R^d)}\leq e^K\Vert 1-e^{\phi}\Vert_{L^1(\R^d)}+e^K\Vert \partial_{x_1}\phi e^{-\phi}\Vert_{L^1(\R^d)}<\infty
\end{eqnarray*}
by (I), (D$\mbox{L}^1$) and \cite[Lemma 7.6]{GiTr77}, i.e.~$\tilde\phi_{\sup}^+\wedge K\in L^1(\R^d)$ holds for any $K\in\N$. In particular, by \cite[Theorem 4.1]{KK02} we have that $\langle\tilde\phi_{\sup}^+\wedge K,\cdot\rangle$ is finite $\mu$-a.s. It follows that for any $K\in\N$ we have
\begin{multline*}
\limsup_{N\to\infty} \int_{\Gamma} 1_{\{\langle\tilde\phi_{\sup}^+\wedge N,\cdot\rangle> N\}}\,d\mu\leq \mu(A_K)+\limsup_{N\to\infty} \int_{\Gamma\setminus A_K} 1_{\{\langle\tilde\phi_{\sup}^+\wedge N,\cdot\rangle> N\}}d\mu\\
\leq \mu(A_K)+\limsup_{N\to\infty} \int_{\Gamma\setminus A_K} 1_{\{\langle\tilde\phi_{\sup}^+\wedge K,\cdot\rangle> N\}}d\mu=\mu(A_K),
\end{multline*}
where $A_K:=\{\gamma\in\Gamma\,|\,\gamma\cap \{\tilde\phi^+_{\sup}>K\}\neq\emptyset\}$. Since $\tilde\phi_{\sup}^+$ is finite Lebesgue-a.e., it follows $\mu(\bigcap_{K\in\N}A_K)=0$. Hence $\mu(A_K)\to 0$ as $K\to\infty$. We conclude that
\begin{equation}\label{eqn:phi+geqN}
1_{\{\langle \tilde\phi_{\sup}^+\wedge N,\cdot\rangle\leq N\}}\to 1\quad \mbox{in $L^1(\Gamma;\mu)$}
\end{equation}
as $N\to\infty$.

By the definition of $\mathcal E_{fro}^{\Gamma,\mu}$ there exists a sequence $(F_n)_{n\in\N}\subset \mathcal FC_b^\infty(\mathcal D,\Gamma)$ such that $F_n\to F$ w.r.t.~the norm $\sqrt{(\cdot,\cdot)_{L^2(\Gamma;\mu)}+\mathcal E_{fro}^{\Gamma,\mu}(\cdot,\cdot)}$. (In particular $\nabla_\gamma^\Gamma F_n\to 0$ in $L^2(\Gamma;\mu)$.) Moreover, by \cite[Theorem I.4.12]{MaRo92} we may assume that $(F_n)_{n\in\N}$ is uniformly bounded in $L^\infty(\Gamma;\mu)$, hence also in $L^\infty(\Gamma;\mu_0)$, and converges also pointwise $\mu_0$-a.s.~to $F$. Since for any $\gamma\in\Gamma$ and $n\in\N$ it holds $F_n(\gamma+e_1)-F_n(\gamma)=\int_0^1 (e_1,\nabla_\gamma^\Gamma F_n(\gamma+se_1))_{\R^d}\,ds$, we find that for any $n,N\in\N$ and $G\in\mathcal FC_b^\infty(\mathcal D,\Gamma)$
\begin{eqnarray}\label{eqn:urghestimate}
\lefteqn{\left\vert \int_{\Gamma} 1_{\{\langle \tilde\phi^+_{\sup}\wedge N,\cdot\rangle\leq N\}}(\gamma)G(\gamma) (F_n(\gamma+e_1)-F_n(\gamma))\,d\mu_0(\gamma)\right\vert}\nonumber\\
& &\leq \int_0^1 \left\vert \int_{\Gamma} G(\gamma-s e_1)1_{\{\langle \tilde\phi^+_{\sup}\wedge N,\cdot\rangle\leq N\}} (\gamma-s e_1) (e_1,\nabla_\gamma^\Gamma F_n(\gamma))_{\R^d}d\mu_0 \right\vert \,ds\nonumber\\
& &\leq e^N\int_0^1 \left\vert \int_\Gamma G(\gamma-s e_1) (e_1,\nabla_\gamma^\Gamma F_n(\gamma))_{\R^d}d\mu\right\vert\,ds.
\end{eqnarray}
For the first inequality we use the translation invariance of $\mu_0$, and for the second inequality we use the fact that for $\gamma\in\Gamma$ fulfilling $\langle \tilde \phi^+\wedge N,\gamma\rangle\leq N$ we automatically have $\tilde\phi^+(x)\leq N$ for all $x\in\gamma$ and hence $\langle \tilde\phi^+,\gamma\rangle=\langle \tilde\phi^+\wedge N,\gamma\rangle\leq N$. Letting $n\to\infty$ in \eqref{eqn:urghestimate} and using the properties of $(F_n)_{n\in\N}$, we find that
$$
\int_{\Gamma} 1_{\{\langle \tilde \phi^+_{\sup}\wedge N,\cdot\rangle \leq N\}}(\gamma) G(\gamma) F(\gamma+e_1)\,d\mu_0(\gamma)=\int_{\Gamma} 1_{\{\langle \tilde \phi^+_{\sup}\wedge N,\cdot\rangle \leq N\}}(\gamma) G(\gamma) F(\gamma)\,d\mu_0(\gamma).
$$
Finally, we let $N\to\infty$ and take \eqref{eqn:phi+geqN} into account to obtain $\int_{\Gamma} G(\gamma)F(\gamma+e_1)d\mu_0(\gamma)=\int_{\Gamma} G(\gamma)F(\gamma)d\mu_0(\gamma)$, which by denseness of $\mathcal FC_b^\infty(\mathcal D,\Gamma)$ in $L^1(\Gamma;\mu_0)$ implies that $F(\cdot+e_1)=F$ holds $\mu_0$-a.s.
\end{proof}

\section{An invariance principle for the tagged particle process}\label{sec:last}

In this section we apply the main result from \cite{DeM89}, which (essentially, see Remark \ref{rem:dmfgw} below) reads as follows:

\begin{theorem}[De Masi, Ferrari, Goldstein, Wick]
Let $(\Omega,\mathcal F,(Y_t)_{t\geq 0},(\mathbf P_y)_{y\in \mathcal X})$ be a Markov process with (measurable) state space $\mathcal X$ which has an invariant and ergodic measure $\nu$ and coordinate maps $\Omega\times [0,\infty)\ni (\omega,t)\mapsto Y_t(\omega)\in \mathcal X$ that are jointly measurable. Let $(X_{[a,b]})_{[a,b]\subset [0,\infty)}$ be a family of random variables on $\Omega$, indexed by all closed bounded subintervals $[a,b]\subset [0,\infty)$ and having the following properties:
\begin{enumerate}
\item $X_{[a,b]}$ is $\mathbf P_\nu$-a.s.~$\sigma(Y_s:\,a\leq s\leq b)$-measurable. (I.e.~$X_{[a,b]}$ is measurable w.r.t.~the smallest $\sigma$-field containing $\sigma(Y_s:\,a\leq s\leq b)$ and all $\mathbf P_\nu$-null sets from $\mathcal F$.)
\item $X_{[a,b]}\in L^1(\mathbf P_\nu)$ for all intervals $[a,b]\subset [0,\infty)$.
\item $X_{[a,b]}$ is antisymmetric w.r.t.~time-reversal and covariant w.r.t.~time-shift.
\item $X_{[a,b]}+X_{[b,c]}=X_{[a,c]}$ $\mathbf P_\nu$-a.s.~for $0\leq a\leq b\leq c$.
\item $(X_t)_{t\geq 0}:=(X_{[0,t]})_{t\geq 0}$ is c\` adl\` ag.
\end{enumerate}

Moreover, assume that the mean forward velocity
\begin{equation}\label{eqn:meanforward}
\lim_{\delta\to 0} \frac{1}{\delta}\mathbf P_{\mu}(X_\delta\,|\,\sigma(Y_0))=\varphi
\end{equation}
exists as a limit in $L^1(\mathcal X;\nu)$ (the conditional expectation in the above formula is interpreted as a function of the initial value, i.e.~a function on $\mathcal X$), and that the martingale
$$
M_t=X_t-\int_0^t\varphi(Y_s)\,ds,\quad t\geq 0,
$$
is square-integrable. For $\varepsilon>0$ and $t\geq 0$ one defines $X_t^\varepsilon:=\varepsilon X_{\varepsilon^{-2}t}$, $X^\varepsilon:=(X_t^\varepsilon)_{t\geq 0}$. Then $X^\varepsilon$ converges in finite dimensional distribution towards an $\R^d$-dimensional Brownian motion scaled by a constant diffusion matrix $D=(D_{ij})_{1\leq i,j\leq d}$. If additionally $\varphi\in L^2(\mathcal X;\nu)$, then this convergence holds weakly in $\nu$-measure.
\end{theorem}

\begin{remark}\label{rem:dmfgw}
The statement of the above theorem is a simplified version of the theorem from \cite{DeM89}. In fact, in that paper also results are derived for the case when the mean forword velocity is only in $L^1$. Moreover, in \cite{DeM89} it is assumed that $X_{[a,b]}$ is even $\sigma(Y_s:\,a\leq s\leq b)$-measurable (i.e.~measurable w.r.t.~the non-completed $\sigma$-fields). However, this assumption can be relaxed to $\mathbf P_\nu$-a.s.~measurability (as we do in assumption (i) of the above theorem) without making any changes to the proof of the theorem.
\end{remark}

\begin{remark}
\begin{enumerate}
\item Weak convergence in $\nu$-measure means in the above theorem that for any bounded continuous function $F$ on the space $D([0,\infty);\R^d)$ of c\` adl\`ag paths (equipped with the Skorokhod topology) and any $\delta>0$ it holds
$$
\int \vert \mathbf E_y[F((\varepsilon X_{\varepsilon^{-2}t})_{t\geq 0})]-\mathbf E[F]\vert>\delta\,d\mu(y)\to 0,
$$
where $\mathbf E$ is the expectation corresponding to the law of the limiting process on $D([0,\infty);\R^d)$. Note that weak convergence in $\nu$-measure implies weak convergence of the laws of $(\varepsilon X_{\varepsilon^{-2}t})_{t\geq 0}$ under $\mathbf P_\nu$.
\item Using an analogon of the portmanteau theorem for weak convergence in measure (see (i)) and recalling that the Skorokhod topology coincides on the space $C([0,\infty);\R^d)$ with the topology of locally uniform convergence and that this space is closed in $D([0,\infty);\R^d)$ (for the latter two facts see \cite[Exercise 3.25]{EK86}), one finds that if $(X_t)_{t\geq 0}$ is not only c\` adl\` ag but even continuous, then from weak convergence in $\nu$-measure it even follows that one has the analogous convergence with $D([0,\infty);\R^d)$ replaced by $C([0,\infty);\R^d)$.
\end{enumerate}
\end{remark}

It is our aim to apply the above theorem for deriving the limit after diffusive scaling of the displacement $(\xi_t-\xi_0)_{t\geq 0}$ of the tagged particle distributed according to $\mathbf P_{h\hat\mu}^{coup}$ (for the definition of this law see Section \ref{sec:displacement}), under the assumptions from Section \ref{sec:conditions} and for $\mu\in \textnormal{ex}\,\mathcal G_{ipb}^{gc}(\Phi_\phi,z\sigma_\phi)$. As we have seen in Theorem \ref{thm:displacement}, we may equivalently consider the limit of the random variables $(\mathbf \eta_t)_{t\geq 0}$ on $C([0,\infty),\ddot\Gamma)$. Using also Lemma \ref{thm:ergodicity}, we find that the Markov process $\mathbf M^{env}$ fulfills the assumptions of the above theorem, and it is not difficult to prove properties (i)-(v) for $X_{[a,b]}$ being the version of $\eta_b-\eta_a$ from Theorem \ref{thm:displacement}(ii), $a,b\in [0,\infty)$, $a<b$. Moreover, we already know from Theorem \ref{thm:martingallemma}(ii) that $X_t-\int_0^t\langle\nabla\phi,\gamma_s\rangle\,ds$, $t\geq 0$, is $\sqrt{2}$ times a Brownian motion and hence in particular a square-integrable martingale. It remains to prove that the mean forward velocity exists and equals $\langle \nabla\phi,\cdot\rangle$.

\begin{lemma}\label{lem:meanforward}
\eqref{eqn:meanforward} holds in the above described setting with $\varphi=\langle \nabla\phi,\cdot\rangle$.
\end{lemma}
\begin{proof}
Observe that $\mathbf E_{\mu}^{env}[\langle \nabla\phi,\gamma_s\rangle\,|\,\sigma(\gamma_0)]=T_{s}^{env}\langle\nabla\phi,\cdot\rangle$ holds $\mu$-a.e.~for any $s\geq 0$. Using Theorem \ref{thm:martingallemma}(ii) it follows that for $\varphi=\langle\nabla\phi,\cdot\rangle$
\begin{align*}
\left\Vert \frac{1}{\delta}\mathbf E_{\mu}^{env}[X_\delta\,|\,\sigma(\gamma_0)]-\varphi\right\Vert_{L^1(\Gamma;\mu)}&=\mathbf E_{\mu}^{env}\left[\left\vert \frac{1}{\delta}\mathbf E_{\mu}^{env}[X_\delta\,|\,\sigma(\gamma_0)]-\varphi(\gamma_0)\right\vert\right]\\
&=\mathbf E_{\mu}^{env}\left[ \frac{1}{\delta} \left\vert \mathbf E_{\mu}^{env}\left[\int_0^\delta \varphi(\gamma_s)-\varphi(\gamma_0)\,ds\bigg|\sigma(\gamma_0)\right]\right\vert\right]\\
&\leq \frac{1}{\delta}\int_0^\delta \mathbf E_{\mu}^{env}\left[\vert T_{s,1}^{env}\varphi(\gamma_0)-\varphi(\gamma_0)\vert\right]\,ds\\&=\frac{1}{\delta}\int_0^\delta \Vert T_{s,1}^{env}\varphi-\varphi\Vert_{L^1(\Gamma;\mu)}\,ds\to 0
\end{align*}
as $\delta\to 0$ by the strong continuity of $(T_{t,1}^{env})_{t\geq 0}$.
\end{proof}

From the above considerations we conclude our final theorem.
\begin{theorem}
Suppose that $\phi$ fulfills the assumptions (SS), (LR), (I) and (D$\mbox{L}^p$) for some $p\in (d,\infty)\cap [2,\infty)$ and that $\mu\in\textnormal{ex}\,\mathcal G_{ipb}^{gc}(\Phi_\phi,z\,\sigma_\phi)$ and assume that $d\geq 2$ or $\phi$ is bounded. Let $0\leq h\in L^1(\R^d)\cap L^\infty(\R^d)$ be a probability density w.r.t.~Lebesgue measure and define $h\hat\mu:=(hd\xi)\otimes\mu$. Denote by $(\mathbf X^{tag})_{t\geq 0}$ the first component of $(\mathbf X^{coup}_t)_{t\geq 0}$. Then, as $\varepsilon\to 0$, $(\varepsilon \mathbf X_{\varepsilon^{-2}t}^{tag})_{t\geq 0}$ converge weakly in $\mu$-measure (in $C([0,\infty),\R^d)$) to a $d$-dimensional Brownian motion scaled with a (constant) diffusion matrix $D=(D_{ij})_{1\leq i,j\leq d}$ and starting in $0$.
\end{theorem}

\section*{Appendix}
\setcounter{theorem}{0}
\def\thetheorem{A\arabic{theorem}}
Here we give a refinement of \cite[Theorem 3.3]{FaGr08} enabling us to drop the condition (LS) on $\phi$ which is assumed there. In contrast to the situation in \cite{AKR98b} this is possible since we are only dealing with grand canonical (in contrast to canonical) Gibbs measures. We remark that an integration by parts formula for grand canonical Gibbs measures (with intensity measure $z\,dx$) has been shown without condition (LS) already in \cite{Yo96}. The present proof has the advantage that it makes use only of the improved Ruelle bound for tempered grand canonical Gibbs measures instead of deriving the necessary estimates by an approximation with finite dimensional systems with empty boundary conditions, and, in our opinion, is much more transparent than the proof given there. (Moreover, it becomes clear that it works for any tempered grand canonical Gibbs measures, not only for those for which the approximation with the mentioned finite dimensional systems holds.) Recall that we defined $B_v^{\phi,\mu}(\gamma):=\langle \textnormal{div}\,v,\gamma\rangle-\sum_{x\in\gamma} (\nabla\phi(x),v(x))_{\R^d}-\sum_{\{x,y\}\subset\gamma} (\nabla\phi(x-y),v(x)-v(y))_{\R^d}$, $\gamma\in\Gamma$. This expression is well-defined in the sense of $\mu$-a.e.~absolute convergence of the sums and integrable w.r.t.~$\mu$ by \cite[Theorem 4.1]{KK02} and the improved Ruelle bound.

\begin{theorem}\label{thm:intbyparts}
Let $\phi$ fulfill (SS), (LR), (I), (D$\mbox{L}^1$) and let $z>0$. Then for any $\mu\in \mathcal G_t^{gc}(\Phi_\phi,ze^{-\phi} \,dx)$, $F\in\mathcal FC_b^\infty(\mathcal D,\Gamma)$ and $v\in C_0^\infty(\R^d)$ it holds
$$
\int_{\Gamma} \nabla_v^\Gamma F\,d\mu=-\int_\Gamma F B_v^{\phi,\mu}\,d\mu.
$$
\end{theorem}

\begin{remark}
\begin{enumerate}
\item Note that the above theorem implies \eqref{eqn:ibp1} (replace $F$ by $FG$ and use the product rule). 
\item It is also possible to derive analogous results for tempered Gibbs measures $\mu\in\mathcal G_t^{gc}(\Phi_\phi,z\rho\,dx)$, when $\rho$ is a bounded function from $W^{1,1}(\R^d)$, if in the definition of $B_v^{\phi,\mu}$ one replaces $\sum_{x\in\gamma} (\nabla\phi(x),v(x))_{\R^d}$ by $\sum_{x\in\gamma} \left(-\frac{\nabla\rho}{\rho}(x),v(x)\right)_{\R^d}$, where one sets $\frac{\nabla\rho}{\rho}(x)=0$ for $\rho(x)=0$. (The proof of the main Lemma \ref{lem:ibpprep} below remains exactly the same, since all such measures $\mu$ fulfill a Ruelle bound).
\item Condition (D$\mbox{L}^1$) can be weakend to just assuming that $e^{-\phi}$ is weakly differentiable and $\nabla e^{-\phi}\in L^1(\R^d)$. One then sets $\nabla\phi:= e^{\phi} 1_{\{\phi<\infty\}} \nabla e^{-\phi}$ and observes that $\nabla e^{-\phi}=- e^{-\phi}\nabla\phi$ Lebesgue-a.e.
\end{enumerate}
\end{remark}

The main work for proving the above theorem is contained in the following lemma.
\begin{lemma}\label{lem:ibpprep}
Let $\phi$ and $\mu$ be as in the above theorem. There is a set $\mathcal M\subset \Gamma$ with $\mu$-measure $0$ such that for all $\gamma\in\Gamma\setminus \mathcal M$ the function
$$
W_\phi(\cdot,\gamma):=\sum_{x\in\gamma} \phi(\cdot-x),
$$
exists - in the sense of convergence of the negative parts of the sum - as a locally bounded below function $\R^d\to\R\cup\{\infty\}$ and is almost everywhere finite, and that $e^{-W_\phi(\cdot,\gamma)}\in W_{loc}^{1,1}(\R^d)$, and the derivative is given by $\nabla e^{-W_\phi(\cdot,\gamma)}\sum_{x\in\gamma} \nabla\phi(x-\cdot)$ a.e.~on $\R^d$, where the sum converges absolutely Lebesgue-a.e.
\end{lemma}
\begin{proof}
We compute for any open relatively compact $\Lambda\subset\R^d$ using the Ruelle bound (which holds for $\mu$, as mentioned in Section \ref{sec:conditions})
$$
\int_\Gamma \int_\Lambda \sum_{x\in\gamma} \vert (\phi\wedge 1)(y-x)\vert\,dy\,d\mu(\gamma)\leq C\int_{\R^d}\int_\Lambda \vert (\phi\wedge 1)\vert(y-x)\,dy\,dx\leq C\vert\Lambda\vert\,\Vert \phi\wedge 1\Vert_{L^1(\R^d)}<\infty,
$$
for some $C<\infty$, where $\vert\Lambda\vert$ denotes the volume of $\Lambda$. This implies that $\mu$-a.e.~we have $dx$-a.e.~finiteness of $W_{\phi\wedge 1}$. Moreover, by (I) we have
$$
\int_\Gamma\int_\Lambda \sum_{x\in\gamma} 1_{\{\phi(y-x)\geq 1\}}\,dy\,d\mu(\gamma)\leq C\vert \Lambda\vert \cdot \vert \{x\in\R^d\,|\,\phi(x)\geq 1\}\vert<\infty,
$$
which implies that for $\mu$-a.e.~$\gamma$ we have $\int_\Lambda \sharp\{x\in\gamma\,|\,\phi(y-x)\geq 1\}\,dy<\infty$, so Lebesgue-a.e.~we have only finitely many particles in the region which makes the difference between $W_{\phi\wedge 1}$ and $W_\phi$. Thus and since $\mu$-a.s.~we have $\phi(x)<\infty$ for all $x\in \gamma$ (because $\phi<\infty$ Lebesgue-a.e.), it follows that for $\mu$-a.e.~$\gamma$ we have Lebesgue-a.e.~finiteness of $W_\phi(\cdot,\gamma)$. Moreover, we have using the Ruelle bound and (LR)
$$
\int_\Gamma \sup_{y\in\Lambda} W_{\phi^-}(y,\gamma)\,dy\leq \int_\Gamma \left\langle \sup_{y\in\Lambda} \psi(\vert y-\cdot\vert_2),\gamma\right\rangle\,d\mu(\gamma)\leq C \int_{\R^d} \sup_{y\in\Lambda}\psi(\vert y-x\vert_2)\,dx,
$$
which is finite because $\psi$ is decreasing. It follows that $\mu$-a.s.~we have that $W_{\phi^-}(\cdot,\gamma)$ is locally bounded and thus $W_{\phi}(\cdot,\gamma)$ is locally bounded from below.

Now we prove the weak differentiability statement, and we may restrict to $\Lambda\subset\R^d$ as above. At first let us have a look at what we consider to be the weak derivative. It holds for $A_M:=\{\gamma\in\Gamma\,|\,W_{\phi^-}(\cdot,\gamma)\leq M\,\mbox{on $\Lambda$}\}$
\begin{multline}\label{eqn:beschrankt}
\int_{A_M}\int_\Lambda \sum_{x\in\gamma} \vert\nabla\phi(y-x)\vert e^{-W_\phi(y,\gamma)}\,dy\,d\mu(\gamma)\leq \int_{A_M} \int_\Lambda \sum_{x\in\gamma} \vert\nabla e^{-\phi(y-x)}\vert e^{W_{\phi^-}(y,\gamma\setminus\{x\})}\,dy\,d\mu(\gamma)\\
\leq e^M \int_\Gamma \int_\Lambda \left\langle \vert \nabla e^{-\phi(\cdot-y)}\vert,\gamma\right\rangle\,dy\,d\mu(\gamma)\leq C\,e^M \vert \Lambda\vert \int_{\R^d} \vert\nabla e^{-\phi}\vert\,dx<\infty,
\end{multline}
which implies existence for $\mu$-a.e.~$\gamma$ of what we consider to be the weak derivative in the sense of Lebesgue-a.e.~convergence of the sum.

Clearly, the function $e^{-W_\phi(\cdot,\gamma_{B_n(0)})}$ is weakly differentiable for all $n$ and $\gamma$, where $\gamma_{B_n(0)}:=\gamma\cap B_n(0)$. The weak derivative is given by
$$
\sum_{x\in \gamma_{B_n(0)}}-(\nabla\phi(\cdot-x)) e^{-W_\phi(\cdot,\gamma_{B_n(0)})}.
$$
We need to prove that $\mu$-a.e.~this converges to the above expression in $L^1(\Lambda)$, at least for a subsequence; therefore, it suffices to prove convergence in $L^1(\Lambda\times\Gamma;dx\otimes\mu)$. We may restrict to $A_M$:
\begin{eqnarray}\label{eqn:dadeldu}
\lefteqn{\int_{A_M}\left\vert\int_\Lambda \sum_{x\in\gamma} \nabla\phi(y-x) e^{-W_\phi(y,\gamma)}-\sum_{x\in\gamma_{B_n(0)}} \nabla\phi(y-x) e^{-W_\phi(y,\gamma_{B_n(0)})}\right\vert\,dy\,d\mu(\gamma)}\nonumber\\
& &\leq \int_{A_M} \int_\Lambda \left\vert \sum_{x\in\gamma\setminus{B_n(0)}} \nabla\phi(y-x) e^{-W_\phi(y,\gamma)}\right\vert\,dy\,d\mu(\gamma)\nonumber\\
& &\quad\quad +\int_{A_M} \int_\Lambda \sum_{x\in \gamma_{B_n(0)}} \vert \nabla(e^{-\phi(y-x)})\vert\,\left\vert e^{-W_\phi(y,\gamma\setminus\{x\})}-e^{-W_\phi(y,\gamma_{B_n(0)}\setminus \{x\})}\right\vert\,dy\,d\mu(\gamma).\quad\quad
\end{eqnarray}
For the first summand we find by an estimate as in \eqref{eqn:beschrankt} that it converges to $0$. For the second one we write
$$
\int_{A_M} \int_\Lambda \sum_{x\in \gamma_{B_n(0)}} \vert \nabla(e^{-\phi(y-x)})\vert\,\left\vert e^{-W_\phi(y,\gamma\setminus\{x\})}-e^{-W_\phi(y,\gamma_{B_n(0)}\setminus \{x\})}\right\vert\,dy\,d\mu(\gamma)=\int \theta_n\,d\mu^*(y,x,\gamma),
$$
where $\mu^*(A):=\int_\Lambda \int_\Gamma \sum_{x\in\gamma} 1_A(y,x,\gamma)\,d\mu(\gamma)dy$, $A\in\mathcal B(\R^d)\otimes \mathcal B(\R^d)\otimes \mathcal B(\Gamma)$ and
$$
\theta_n(y,x,\gamma):=\vert \nabla e^{-\phi(y-x)}\vert 1_{B_n(0)}(x)\left\vert e^{-W_\phi(y,\gamma\setminus\{x\})}-e^{-W_\phi(y,\gamma_{B_n(0)}\setminus\{x\})}\right\vert,\quad y\in\Lambda, \gamma\in \Gamma, x\in \gamma.
$$
Now we know that for a.e.~$\gamma$ and $y$ we have that $\phi(y-x)<\infty$ and $\vert \nabla e^{-\phi(y-x)}\vert<\infty$ for all $x\in\gamma$ and $W_\phi(y,\gamma_{B_n(0)})\setminus\{x\})\to W_\phi(y,\gamma\setminus \{x\})$. Hence we have $\mu^*$-a.s.~convergence to $0$ of the $\theta_n$. Moreover, for $\gamma\in A_M$, $y\in\Lambda$ and $x\in\gamma$ we have
$$
\theta_n(y,x,\gamma)\leq 2 \vert \nabla e^{-\phi(y-x)}\vert e^M,
$$
but
$$
\int \vert \nabla e^{-\phi(y-x)}\vert \,d\mu^*(y,x,\gamma)=\int_\Lambda\int_\Gamma \left\langle \vert\nabla e^{-\phi(y-\cdot)}\vert,\gamma\right\rangle\,d\mu(\gamma)\,dy\leq C\vert \Lambda\vert\,\Vert \nabla e^{-\phi}\Vert_{L^1(\R^d)}<\infty,
$$
hence by the dominated convergence theorem we obtain $\theta_n\to 0$ in $L^1(\Lambda\times\R^d\times\Gamma;\mu^*)$, and thus convergence of the left-hand side of \eqref{eqn:dadeldu} to $0$ as $n\to\infty$. It follows that for a sequence $(n_k)_{k\in\N}\subset \N$ we have for $\mu$-a.e.~$\gamma\in A_M$ that $\sum_{x\in\gamma_{B_{n_k}(0)}} \nabla\phi(\cdot-x) e^{-W_\phi(\cdot,\gamma_{B_{n_k}(0)})}\to \sum_{x\in\gamma} \nabla\phi(\cdot-x) e^{-W_\phi(\cdot,\gamma)}$ in $L^1(\Lambda)$. Thus the mentioned weak differentiability holds for $\mu$-a.e.~$\gamma\in A_M$, and taking the union over all $M\in\N$, the assertion follows.
\end{proof}

\begin{proof}[Proof of Theorem \ref{thm:intbyparts}]

Choose $\Lambda\subset\R^d$ open, relatively compact and large enough such that $F$ does not depend on the configuration outside $\Lambda$, and the support of $v$ is contained in $\Lambda$. We set 
$$
U_{\phi}^{\gamma,\Lambda}(\{x_1,\cdots,x_n\}):=\sum_{i<j}\phi(x_i-x_j)+\sum_{i=1}^n W_\phi(x_i,\gamma\setminus \Lambda)+\sum_{i=1}^n \phi(x_i)
$$ 
for $x_1,\cdots,x_n\in \Lambda$ and $\gamma\in\Gamma$. The Dobrushin-Lanford-Ruelle equation implies that
\begin{equation}\label{eqn:DLR}
\int_\Gamma \nabla_v^\Gamma F\,d\mu=\frac{1}{Z}\int_\Gamma\sum_{n=0}^\infty \frac{z^n}{n!}\int_{\Lambda^n} \sum_{i=1}^n v(x_i)\nabla_{x_i}F(\{x_1,\cdots,x_n\}) e^{-U_{\phi}^{\gamma,\Lambda}(\{x_1,\cdots,x_n\})}\,dx_1\cdots dx_n\,d\mu(\gamma),
\end{equation}
where $Z$ is a normalization constant. Due to Lemma \ref{lem:ibpprep} we may apply integration by parts for all $\gamma\in\Gamma\setminus \mathcal M$ (and thus for $\mu$-a.e.~$\gamma$) to the inner integral on the right-hand side, which therefore equals
\begin{multline*}
-\sum_{i=1}^n \int_{\Lambda^n} F(\{x_1,\cdots,x_n\})\,e^{-U_\phi^{\gamma,\Lambda}(\{x_1,\cdots,x_n\},\gamma)}\bigg(\textnormal{div}\,v(x_i)-v(x_i) \nabla\phi(x_i)\\-\sum_{j\neq i} v(x_i)\nabla\phi(x_i-x_j)-\sum_{y\in\gamma} v(x_i)\nabla\phi(x_i-y)\bigg)dx_1\cdots dx_n.
\end{multline*}
Taking the sum inside the brackets in the above expression and summing up the first summand in the bracket yields $\langle \textnormal{div}\,v,\{x_1,\cdots,x_n\}\cup\gamma\setminus \Lambda\rangle$, the second gives $\langle v\nabla\phi,\{x_1,\cdots,x_n\}\cup\Gamma\setminus\Lambda)$ from the last two we obtain $\sum_{\{z,z'\}\subset \{x_1,\cdots,x_n\}\cup \gamma\setminus\Lambda} \nabla\phi(z-z')(v(z)-v(z'))$. Therefore, the above expression equals
$$
-\int_{\Lambda^n} F(\{x_1,\cdots,x_n\})\,e^{-U_\phi^{\gamma,\Lambda}(\{x_1,\cdots,x_n\},\gamma)} B_{v}^{\phi,\mu}(\{x_1,\cdots,x_n\}\cup \gamma\setminus\Lambda)\,dx_1\cdots dx_n.
$$
Replacing the inner integral in \eqref{eqn:DLR} by this expression and noting that the sum over $n$ exists by integrability of $B_v^{\phi,\mu}$ w.r.t.~$\mu$, we obtain the assertion.
\end{proof}

\vspace{2ex}\noindent{\textbf{Acknowledgement:}} We would like to thank Yuri Kondratiev, Tobias Kuna, Eugene Lytvy\-nov, Michael R\"ockner and Sven Struckmeier for helpful discussions and remarks. Financial support by the DFG via project GR 1809/8-1 is gratefully acknowledged.

\end{document}